\documentclass[11pt]{article}%
\usepackage{amsmath}
\usepackage{amssymb}
\usepackage{amsfonts}

\usepackage{cite}
\usepackage{graphicx}
\usepackage{float}
\usepackage{longtable}
\usepackage[utf8]{inputenc}
\usepackage{hyperref}

\usepackage{algorithmicx}
\usepackage[ruled,vlined]{algorithm2e}
\SetAlgoCaptionSeparator{ }

\usepackage{fullpage}

\newcommand{\fref}[1]{Fig.~\ref{#1}}

\newcommand{\tref}[1]{Table~\ref{#1}}
\newcommand{\sref}[1]{Section~\ref{#1}}

\newenvironment{proced}[1][!htbp]
  {
   \begin{algorithm}[#1]%
  }{\end{algorithm}}
\providecommand{\U}[1]{\protect\rule{.1in}{.1in}}
\newtheorem{theorem}{Theorem}

\newtheorem{problem}{Problem}

\newenvironment{proof}[1][Proof]{\textbf{#1.} }{\ \rule{0.5em}{0.5em}}

\begin{document}

\title{Dense Quantum Measurement Theory}
\author{Laszlo Gyongyosi\thanks{School of Electronics and Computer Science, University of Southampton, Southampton SO17 1BJ, U.K., and Department of Networked Systems and Services, Budapest University of Technology and Economics, 1117 Budapest, Hungary, and MTA-BME Information Systems Research Group, Hungarian Academy of Sciences, 1051 Budapest, Hungary.}
\and Sandor Imre\thanks{Department of Networked Systems and Services, Budapest University of Technology and Economics, 1117 Budapest, Hungary.}}

\date{}

\maketitle
\begin{abstract}
Quantum measurement is a fundamental cornerstone of experimental quantum computations. The main issues in current quantum measurement strategies are the high number of measurement rounds to determine a global optimal measurement output and the low success probability of finding a global optimal measurement output. Each measurement round requires preparing the quantum system and applying quantum operations and measurements with high-precision control in the physical layer. These issues result in extremely high-cost measurements with a low probability of success at the end of the measurement rounds. Here, we define a novel measurement for quantum computations called dense quantum measurement. The dense measurement strategy aims at fixing the main drawbacks of standard quantum measurements by achieving a significant reduction in the number of necessary measurement rounds and by radically improving the success probabilities of finding global optimal outputs. We provide application scenarios for quantum circuits with arbitrary unitary sequences, and prove that dense measurement theory provides an experimentally implementable solution for gate-model quantum computer architectures.
\end{abstract}

\section{Introduction}
\label{sec1}
Quantum measurement is a crucial subject in quantum computation and communication \cite{ref1,ref2,ref3,ref4,ref5,ref6,ref7,ref8,ref9,ref10,ref11,ref12,ref13,ref14,ref15,ref16,ref17,ref18,ref19,ref24,ref26, refibm,refpr,refha,aar,ref28,ref29,ref30,ref31,ref32,ref33}. The aim of quantum measurement is to extract valuable and useable information from the measured quantum system. The measurement operator connects the quantum world and our traditional, classical world. While the input of the measurement can be a superposed or entangled quantum system, the output of the measurement is classical information (i.e., bitstrings). Quantum measurements can be performed in different ways, for example via projective \cite{ m1,m2,m3,m3b,m4,m5,m6,m7} or POVM (positive-operator valued measure) measurements \cite{ref25,ref27, m3b,p1,p2,p3,p4,p5}.

Quantum measurement is required element in high-complexity quantum computations, in high-performance quantum information processing and in quantum computer architectures. The main issues of current quantum measurement strategies are the high number of measurement rounds and the probability of successfully finding a global optimal measurement output. The necessity of a high number of measurement rounds requires preparing the input quantum system and applying quantum operations with high-precision control in the physical layer through several rounds, which results in a high-cost procedure overall that is not tractable in any experimental setting. The repetition of a measurement round therefore requires in each round the careful preparation of a quantum register of quantum states that are then fed into a quantum circuit that realizes an arbitrary unitary sequence. In each round, the output of the quantum circuit is measured by a measurement array $M$, which produces a classical output string $z$. The aim is then to find a global optimal output $z^{*} $ that describes the properties of the output quantum system with the highest accuracy according to quality measurement functions. An example of a high-cost application of standard measurement is measuring the output of a quantum circuit applied to realize quantum computations where the quantum circuit is set to perform a unitary operation $U$. Without loss of generality, the $n$-length input quantum system ${\left| X \right\rangle} $ of the quantum circuit is assumed to be a superposed quantum system that is fed into the circuit. Then, the $n$-length output quantum system ${\left| Y \right\rangle} =U{\left| X \right\rangle} $ is measured by the measurement operator $M$, which produces a string $z$ and, after repeating the procedure $R_{0} $ times, yields the global optimal string $z^{*} $ with success probability $\Pr _{R_{0} } \left(z^{*} \right)$. Assuming that $U$ is an arbitrary quantum circuit and $M$ is a standard measurement, the measurement procedure requires high repetition numbers, while the success probability remains low (An example is the application of standard quantum measurements in quantum computers, where for $R_{0} \approx 100$ standard measurement rounds, the achievable success probability is approximately $\Pr _{R_{0} } \left(z^{*} \right)\ge 0.01$ \cite{ref10}). Since each measurement round requires high-cost and high-precision quantum state preparations and quantum operations, the total cost to find the global optimal $z^{*} $ is very high in a practical setting. To avoid the issues of a high number of measurement rounds and the low success probability of quantum measurements, a novel measurement is essential for quantum computations.  

Here, we define a novel measurement for quantum computations called dense quantum measurement. The dense measurement strategy aims at fixing the drawbacks of standard quantum measurements by achieving a radical reduction in the number of necessary measurement rounds and by significantly improving the success probabilities of finding global optimal outputs (see Theorem 1 for the system model). Dense quantum measurement requires only $R{\rm \ll }R_{0} $ measurement rounds, such that $R$ rounds leads to a success probability of $\Pr \left(z^{*} \right){\rm \gg }\Pr _{R_{0} } \left(z^{*} \right)$. The dense measurement strategy is rooted in the theory of compressed sensing \cite{ref20,ref21,ref22,ref23}, which allows recovering noisy signals with a high efficiency in the field of traditional communications. Dense quantum measurement utilizes an $M_{r} $ randomized measurement operator that is defined as an $n$-bit length vector $M_{r} =\left(b_{1} M_{B} ,\ldots ,b_{n} M_{B} \right)^{T} $, where $b_{i} $ is a random variable, $b_{i} \in \left\{0,1\right\}$, ${\rm Pr}\left(0\right)={\rm Pr}\left(1\right)=0.5$, associated with the measurement of the $i$-th quantum state of the output quantum system, while $M_{B} $ is a quantum measurement in the computational basis $B$; thus, $b_{i} M_{B} =0$ if $b_{i} =0$ and $b_{i} M_{B} =M_{B} $ if $b_{i} =1$. As follows, the $M_{B} $ measurement in the computational basis is discarded if $b_{i} =0$. Then, the measurement result is post-processed via unit ${\rm {\mathcal P}}$ that integrates algorithms to determine the global optimal string $z^{*} $ from the results of the randomized measurements.

As we prove (see Theorem 2), the number $R_{0} $ of standard measurement rounds can be reduced to $R=\alpha Z^{2} K\log ^{4} \left(n\right)$ dense measurement rounds for an arbitrary quantum circuit, where $K\ge L_{0} \left(S\right)$ and $K{\rm \ll }n$, while $\alpha >0,Z>0$ are constants. At this number of measurement rounds, the success probability is $\Pr \left(z^{*} \right)=1-n^{-\log ^{3} \left(n\right)} \approx 1$ for any practical value of $n$.
We also prove that if the output of the quantum circuit is a computational basis quantum state, then $R_{0} $ can be reduced to $R={\rm {\mathcal O}}\left(\gamma K\log \left({\textstyle\frac{10n}{K}} \right)\right)$ dense measurement rounds, where $\gamma >0$ is a constant, such that $\Pr \left(z^{*} \right)=1-2\exp \left(-R\right)\approx 1,$ for any $R$ (see Theorem 3).

The novel contributions of our manuscript are as follows:
\begin{enumerate}
\item \textit{We define a novel quantum measurement theory called dense quantum measurement.}
\item \textit{We prove that dense measurement reduces the number of required measurement rounds to find a global optimal output.}
\item \textit{We prove that dense measurement significantly improves the success probability of finding a global optimal output.}
\item \textit{We provide an application scenario for quantum circuits with arbitrary unitary sequences, and for the dense measurement of computational basis quantum states in  gate-model quantum computer environment.} 
\item \textit{We reveal that the primary advantages of dense quantum measurement theory are the significantly lower measurement rounds and significantly higher success probabilities.} 
\end{enumerate}

This manuscript is organized as follows. In \sref{relw}, the related works are summarized. In \sref{sec2}, the problem statement is given. In \sref{sec3}, preliminaries are summarized. \sref{sec4} proposes the theorems and proofs. \sref{sec5} provides a performance evaluation. Finally, \sref{sec6} concludes the paper. Supplemental information is included in the Appendix.

\section{Related Works}
\label{relw}
The related works on quantum measurement theory, gate-model quantum computers and compressed sensing are summarized as follows.
\subsection{Quantum Measurement Theory}
Quantum measurement has a fundamental role in quantum mechanics with several different theoretical interpretations \cite{m1,m2,m3,m3b,m4,m5,m6,m7,p1,p2,p3,p4,p5}. The measurement of a quantum system collapses of the quantum system into an eigenstate of the operator corresponding to the measurement. The measurement of a quantum system produces a measurement result, the expected values of measurement are associated with a particular probability distribution. 

In quantum mechanics several different measurement techniques exist. In a projective measurement \cite{m1,m2,m3,m3b,m4,m5,m6,m7}, the measurement of the quantum system is mathematically interpreted by projectors that project any initial quantum state onto one of the basis states. The projective measurement is also known as von Neumann measurement \cite{m1}. In our manuscript the projective measurement with no post-processing on the measurement results is referred to as \textit{standard measurement}\footnote{It is motivated by the fact, that in a gate-model quantum computer environment the output quantum system is measured with respect to a particular computational basis.}. 

The von Neumann measurements are a special case of a more general measurement, the POVM measurement \cite{m3b,p1,p2,p3,p4,p5}. Without loss of generality, the POVM is a generalized measurement that can be interpreted as a von Neumann measurement that utilizes an additional quantum system (called ancilla). The POVM measurement is mathematically described by a set of positive operators such that their sum is the identity operator \cite{m8,m9,m10}. The POVM measurements therefore can be expressed in terms of projective measurements (see also Neumark's dilation theorem \cite{n1,n2,n3}).

Another subject connected to quantum measurement theory is quantum-state discrimination \cite{u1,u2,u3,u4,u5} that covers the distinguishability of quantum states, and the problem of differentiation between non-orthogonal quantum states. 

\subsection{Gate-Model Quantum Computers}
The theoretical background of the gate-model quantum computer environment utilized in our manuscript can be found in \cite{ref9} and \cite{ref10}.

In \cite{ref10}, the authors studied the subject of objective function evaluation of computational problems fed into a gate-model quantum computer environment. The work focuses on a qubit architectures with a fixed hardware structure in the physical layout. In the system model of a gate-model quantum computer, the quantum computer is modeled as a sequence of unitary operators (quantum gates). The quantum gates are associated with a particular control parameter called the gate parameter. The quantum gates can process one-qubit length and multi-qubit length quantum systems. The input quantum system (particularly a superposed quantum system) of the quantum circuit is transformed via a sequence of unitaries controlled via the gate parameters, and the output qubits are measured by a measurement array. The measurement in the model is realized by a projective measurement applied on a qubits that outputs a logical bit with value zero or one for each measured qubit. The result of the measurement is therefore a classical bitstring. The output bitstring is processed further to estimate the objective function of the quantum computer. 
The work also induces and opens several important optimization questions, such as the optimization of quantum circuits of gate-model quantum computers, optimization of objective function estimation, measurement optimization and optimization of post-processing in a gate-model quantum computer environment. In our particular work we are focusing on the optimization of the measurement phase. 

An optimization algorithm related to gate-model quantum computer architectures is defined in \cite{ref9}. The optimization algorithm is called “Quantum Approximate Optimization Algorithm” (QAOA). The aim of the algorithm is to output approximate solutions for combinatorial optimization problems fed into the quantum computer. The algorithm is implementable via gate-model quantum computers such that the depth of the quantum circuit grows linearly with a particular control parameter. The work also proposed the performance of the algorithm at the utilization of different gate parameter values for the unitaries of the gate-model computer environment.

In \cite{su}, the authors studied some attributes of the QAOA algorithm. The authors showed that the output distribution provided by QAOA cannot be efficiently simulated on any classical device. A comparison with the “Quantum Adiabatic Algorithm” (QADI) \cite{adi1,adi2} is also proposed in the work. The work concluded that the QAOA can be implemented on near-term gate-model quantum computers for optimization problems.

An application of the QAOA algorithm to a bounded occurrence constraint problem “Max E3LIN2” can be found in \cite{ref12}. In the analyzed problem, the input is a set of linear equations each of which has three boolean variables, and each equation outputs whether the sum of the variables is 0 or is 1 in a mod 2 representation. The work is aimed to demonstrate the capabilities of the QAOA algorithm in a gate-model quantum computer environment. 

In \cite{refa3}, the authors studied the objective function value distributions of the QAOA algorithm. The work concluded, at some particular setting and conditions the objective function values could become concentrated. A conclusion of the work, the number of running sequences of the quantum computer can be reduced.

In \cite{refa4}, the authors analyzed the experimental implementation of the QAOA algorithm on near-term gate-model quantum devices. The work also defined an optimization method for the QAOA, and studied the performance of QAOA. As the authors found, the QAOA can learn via optimization to utilize non-adiabatic mechanisms. 

In \cite{refa5}, the authors studied the implementation of QAOA with parallelizable gates. The work introduced a scheme to parallelize the QAOA for arbitrary all-to-all connected problem graphs in a layout of qubits. The proposed method was defined by single qubit operations and the interactions were set by pair-wise CNOT gates among nearest neighbors. As the work concluded, this structure allows for a parallelizable implementation in quantum devices with a square lattice geometry.

In \cite{ref11}, the authors defined a gate-model quantum neural network. The gate-model quantum neural network describes a quantum neural network implemented on gate-model quantum computer. The work focuses on the architectural attributes of a gate-model quantum neural network, and studies the training methods. A particular problem studied in the work is the classification of classical data sets which consist of bitstrings with binary labels. In the architectural model of a gate-model quantum neural network, the weights are represented by the gate parameters of the unitaries of the network, and the training method acts these gate parameters. As the authors stated, the gate-model quantum neural networks represent a practically implementable solution for the realization of quantum neural networks on near-term gate-model quantum computer architectures. 

In \cite{sat}, the authors defined a quantum algorithm that is realized via a quantum Markov process. The analyzed process of the work was a quantum version of a classical probabilistic algorithm for $k$-SAT defined in \cite{sch}. The work also studied the performance of the proposed quantum algorithm and compared it with the classical algorithm.

For a review on the noisy intermediate-scale quantum (NISQ) era and its technological effects and impacts on quantum computing, see \cite{refpr}. 

The subject of quantum computational supremacy (tasks and problems that quantum computers can solve but are beyond the capability of any classical computer) and its practical implications are studied in \cite{refha}. For a work on the complexity-theoretic foundations of quantum supremacy, see \cite{aar}.

A comprehensive survey on quantum channels can be found in \cite{ref19}, while for a survey on quantum computing technology, see \cite{refsur}. 

\subsection{Compressed Sensing}
In traditional information processing, compressed sensing \cite{ref20} is a technique to reduce the sampling rate to recover a signal from fewer samples than it is stated by the Shannon-Nyquist sampling theorem (that states that the sampling rate of a continuous-time signal must be twice its highest frequency for the reconstruction) \cite{ref20,ref21,ref22,ref23}. In the framework of compressed sensing, the signal reconstruction process exploits the sparsity of signals (in the context of compressed sensing, a signal is called sparse if most of its components are zero) \cite{ref23,refcs1,refcs2,refcs3,refcs4,refcs5}. Along with the sparsity, the restricted isometry property \cite{ref23,refcs1,refcs5} is also an important concept of compressed sensing, since, without loss of generality, this property makes it possible to yield unique outputs from the measurements of the sparse inputs. The restricted isometry property is also a well-studied problem in the field of compressed sensing \cite{refcs16,refcs17,refcs18,refcs19,refcs19,refcs20}. 

A special technique within compressed sensing is the so-called “1-bit” compressed sensing \cite{one1,one2,one3}, where 1-bit measurements are applied that preserve only the sign information of the measurements.

The application of compressed sensing covers the fields of traditional signal processing, image processing and several different fields of computational mathematics \cite{refcs6,refcs7,refcs8,refcs9,refcs10,refcs11,refcs12,refcs13}. 

The dense quantum measurement theory proposed in our manuscript also utilizes the fundamental concepts of compressed sensing. However, in our framework the primary aims are the reduction of the measurement rounds required to determine a global optimal output at arbitrary unitaries, and the boosting of the success probability of finding a global optimal output at a particular measurement round. The results are illustrated through a gate-model quantum computer environment.

\section{Problem Statement}
\label{sec2}
Let ${\left| X \right\rangle} $ be the superposed input system of a quantum circuit with a $QG$ quantum gate structure, formulated by $n$ quantum states, as
\begin{equation} \label{ZEqnNum910801} 
{\left| X \right\rangle} ={\textstyle\frac{1}{\sqrt{d^{n} } }} \sum _{z}{\left| z \right\rangle}  ,                                                                   
\end{equation} 
where $d$ is the dimension of the quantum system, ${\left| z \right\rangle} $ is a computational basis state and $U(\vec{\theta })$ is the unitary operation of $QG$, defined as a sequence of $L$ unitaries
\begin{equation} \label{ZEqnNum489626} 
U(\vec{\theta })=U_{L} \left(\theta _{L} \right)U_{L-1} \left(\theta _{L-1} \right),\ldots ,U_{1} \left(\theta _{1} \right),                                              
\end{equation} 
where $\vec{\theta }$ is the $L$-dimensional vector of the gate parameters of the unitaries (gate parameter vector):
\begin{equation} \label{3)} 
\vec{\theta }=\left(\theta _{1} ,\ldots ,\theta _{L} \right)^{T} .                                                                   
\end{equation} 
In \eqref{ZEqnNum489626}, an $i$-th unitary gate $U_{i} \left(\theta _{i} \right)$ is evaluated as
\begin{equation} \label{4)} 
U_{i} \left(\theta _{i} \right)=\exp \left(-i\theta _{i} P\right),                                                               
\end{equation} 
where $P$ is a generalized Pauli operator formulated by the tensor product of Pauli operators $\left\{\sigma _{X} ,\sigma _{Y} ,\sigma _{Z} \right\}$.

In a standard measurement setting, the ${\left| Y \right\rangle} $ output of $QG$ is 
\begin{equation} \label{5)} 
{\left| Y \right\rangle} =U(\vec{\theta }){\left| X \right\rangle}  
\end{equation} 
measured by a $M$ measurement operator, which yields an output string $z$ as
\begin{equation} \label{6)} 
z=M{\left| Y \right\rangle} .                                                                          
\end{equation} 
The global optimal output string $z^{*} $ is an output string that yields the optimal estimation $C\left(z^{*} \right)$ at a particular objective function $C$ fed into the quantum circuit as a maximization problem
\begin{equation} \label{7)} 
C\left(z^{*} \right)=\mathop{\max }\limits_{\forall m} C\left(z_{m} \right),                                                                  
\end{equation} 
where $C\left(z_{m} \right)$ is the estimate yielded in an $m$-th measurement round, $m=1,\ldots ,R_{0} $, while $z_{m} $ is the output string yielded in the $m$-th round.

Without loss of generality, after $R_{0} $ measurement rounds, the probability that the global optimal output string $z^{*} $ is determined is $\Pr _{R_{0} } \left(z^{*} \right)$; thus, $C\left(z^{*} \right)$ can be found with the same success probability, 
\begin{equation} \label{8)} 
{{\Pr }_{R_{0} }} C\left(z^{*} \right)={{\Pr }_{R_{0}} } \left(z^{*} \right).                                                                
\end{equation} 
The problems connected to the general measurement strategy to find $z^{*} $ are the high number of $R_{0} $ repetitions and the low $\Pr _{R_{0} } \left(z^{*} \right)$ success probability. Consequently, the standard measurement procedure requires high-cost quantum state preparations, the application of high-cost measurement arrays and high-precision control and calibrations in the physical layer.

Problems 1-3 summarize the problems to be solved.

\begin{problem} 
(System Model). Define a novel quantum measurement strategy for the significant reduction of the $R_{0} $ measurement rounds of standard measurements and for the significant improvement of the  $\Pr _{R_{0} } \left(z^{*} \right)$ success probability in determining a global optimal output $z^{*} $.
\end{problem}

\begin{problem} 
(General application). Define $R$ and $\Pr \left(z^{*} \right)$ for an arbitrary quantum circuit with $U(\vec{\theta })$. Prove the number $R$ of measurement rounds, $R{\rm \ll }R_{0} $, and the $\Pr \left(z^{*} \right)$ success probability, $\Pr \left(z^{*} \right){\rm \gg }\Pr _{R_{0} } \left(z^{*} \right)$. 
\end{problem} 

\begin{problem}
(Dense measurement of computational basis quantum states). Define $R$ and $\Pr \left(z^{*} \right)$ for an arbitrary quantum circuit with $U(\vec{\theta })=U_{B} $, where $U_{B} $ sets the computational basis $B$ \footnote{Throughout the manuscript, the term ``computational basis'' refers to a basis $B$, for which $L_{0} \left(S\right)\le K$ holds at a given $S=BX$, where $X$ is an input system.}. Prove the number $R$ of measurement rounds, $R{\rm \ll }R_{0} $, and the $\Pr \left(z^{*} \right)$ success probability, $\Pr \left(z^{*} \right){\rm \gg }\Pr _{R_{0} } \left(z^{*} \right)$.
\end{problem}

The resolutions of Problems 1-3 are given in Theorems 1-3, respectively.

\section{Preliminaries}
\label{sec3}
\subsection{Sub-Gaussian Distributions}
\label{subg}

A random variable $X$ is sub-Gaussian, if for the probability distribution of $X$,
\begin{equation} \label{115)} 
\Pr \left(\left|X\right|\ge \kappa \right)\le C_{1} e^{-C_{2} \kappa ^{2} }  
\end{equation} 
holds for $\forall \kappa >0$, where 
\begin{equation} \label{116)} 
C_{1} ,C_{2} >0 
\end{equation} 
are sub-Gaussian parameters. 

By theory, if $X$ is sub-Gaussian with 
\begin{equation} \label{ZEqnNum625894} 
{\rm {\mathbb{E}}}\left(X\right)=0,                                                                        
\end{equation} 
then there exists a constant $c^{*} $ depending on only $C_{1} ,C_{2} $ such that
\begin{equation} \label{ZEqnNum941241} 
{\rm {\mathbb{E}}}\left(\exp \left(\eta X\right)\le \exp \left(c^{*} \eta ^{2} \right)\right) 
\end{equation} 
for $\forall \eta \in {\rm {\rm R}}$. 

If \eqref{ZEqnNum941241} holds, then \eqref{ZEqnNum625894} is satisfied such that the $C_{1} $ sub-Gaussian parameter of $X$ is
\begin{equation} \label{119)} 
C_{1} =2,                                                                            
\end{equation} 
and $C_{2} $ is as
\begin{equation} \label{120)} 
C_{2} ={\textstyle\frac{1}{4c^{*} }} .                                                                             
\end{equation} 
An $M\times N$ random matrix $M$ is a sub-Gaussian random matrix, if 
\begin{equation} \label{121)} 
\Pr \left(\left|M_{j,k} \right|\ge \kappa \right)\le C_{1} e^{-C_{2} \kappa ^{2} }  
\end{equation} 
for $\forall \kappa >0$, where $M_{j,k} $ is the $\left(j,k\right)$-th element of $M$, $j\in \left[M\right],k\in \left[N\right]$, where 
\begin{equation} \label{122)} 
C_{1} ,C_{2} >0 
\end{equation} 
are sub-Gaussian parameters.

\section{Methods}
\label{sec4}
\subsection{System Model}
\begin{theorem} 
 (Dense measurement). A $QG$ structure with unitary $U(\vec{\theta })=U_{B} U(\vec{\theta '})$, where the unitary sets an arbitrary computational basis $B$ for an $n$-length input ${\left| X \right\rangle} $ as $U_{B} {\left| X \right\rangle} ={\left| S \right\rangle} $, such that $L_{0} \left(S\right)\le K$, $K{\rm \ll }n$, holds for the $L_{0} $-norm of $S$, where $S$ is a classical representation of ${\left| S \right\rangle} $, while $U(\vec{\theta '})$ is the actual setting of the unitaries of $QG$ at ${\left| S \right\rangle} $ and with a $M_{r} $ random measurement operator, allows the determination of the global optimal output $z^{*} $ and global optimal estimate $C\left(z^{*} \right)$ at a particular objective function $C$ as $\varepsilon =\Pr \left(\delta _{K} \ge \chi \right)$ holds, where $\delta _{K} $ and $\chi $ are constants depending on ${\rm {\mathcal Q}}={\rm {\mathcal M}}U(\vec{\theta '})$, where ${\rm {\mathcal M}}=\left(M_{r}^{\left(1\right)} ,\ldots ,M_{r}^{\left(R\right)} \right)$ and $M_{r}^{\left(m\right)} $ is the measurement operator of the $m$-th dense measurement round $m=1,\ldots ,R$.
\end{theorem}
\begin{proof}
First, we rewrite \eqref{ZEqnNum489626} as 
\begin{equation} \label{ZEqnNum502902} 
U(\vec{\theta })=U_{B} U(\vec{\theta '}),                                                                      
\end{equation} 
where $U_{B} $ is a unitary that sets a computational basis $B$ and $U(\vec{\theta '})$ is a unitary operation that sets the unitaries, such that
\begin{equation} \label{10)} 
\begin{split}
  & {{U}_{B}}U( {\vec{{\theta }'}}){{( {{U}_{B}}U( {\vec{{\theta }'}}))}^{\dagger }} \\ 
 & ={{U}_{B}}U( {\vec{{\theta }'}})U_{B}^{\dagger }{{( U( {\vec{{\theta }'}}))}^{\dagger }} \\ 
 & ={{U}_{B}}U( {\vec{{\theta }'}}){{( U( {\vec{{\theta }'}}))}^{\dagger }}U_{B}^{\dagger } \\ 
 & ={{U}_{B}}IU_{B}^{\dagger } \\ 
 & ={{U}_{B}}U_{B}^{\dagger }=I,  
\end{split}
\end{equation} 
where $I$ is the identity and $\vec{\theta '}$ is the $L$-dimensional vector of the gate parameters of $U(\vec{\theta '})$. Applying the unitary $U_{B} $ on input system ${\left| X \right\rangle} $ yields the $n$-length quantum system ${\left| S \right\rangle} ={\left| s_{1} ,\ldots ,s_{n}  \right\rangle} $,
\begin{equation} \label{ZEqnNum536134} 
{\left| S \right\rangle} =U_{B} {\left| X \right\rangle} ,                                                                        
\end{equation} 
where the computational basis $B$ for $U_{B} $ in \eqref{ZEqnNum502902} is selected such that for the $L_{0} $-norm of $S$ the following relation holds
\begin{equation} \label{ZEqnNum334843} 
L_{0} \left(S\right)=\left\| S\right\| _{0} \le K,                                                                    
\end{equation} 
where 
\begin{equation} \label{ZEqnNum107628} 
S=BX 
\end{equation} 
is a classical representation of ${\left| S \right\rangle} $, $X$ is a classical representation of ${\left| X \right\rangle} $ and $K{\rm \ll }n$. Therefore, $B$ can be an arbitrary computational basis for which \eqref{ZEqnNum334843} holds at a given \eqref{ZEqnNum107628} (For example, if $B$ is the Fourier basis, then $U_{B} $ realizes a quantum Fourier transform).

The output of $QG$ at \eqref{ZEqnNum502902} and \eqref{ZEqnNum536134} is therefore written as
\begin{equation} \label{ZEqnNum755756} 
\begin{split}
   U( {\vec{\theta }})\left| X \right\rangle &={{U}_{B}}U( {\vec{{\theta }'}})\left| X \right\rangle  \\ 
 & =U( {\vec{{\theta }'}})\left( {{U}_{B}}\left| X \right\rangle  \right) \\ 
 & =U( {\vec{{\theta }'}})\left| S \right\rangle  \\ 
 & =\left| G \right\rangle  \\ 
 & =\left| {{g}_{1}},\ldots ,{{g}_{n}} \right\rangle ,  
\end{split}
\end{equation} 
whose state is measured by an $M_{r} $ random measurement operator, defined as an $n$-bit length vector 
\begin{equation} \label{ZEqnNum688653} 
M_{r} =\left(b_{1} M_{B} ,\ldots ,b_{n} M_{B} \right)^{T} ,                                                                
\end{equation} 
where $b_{i} $ is a random variable,
\begin{equation} \label{ZEqnNum478723} 
b_{i} =\left\{\begin{array}{l} {0,{\rm \; with\; Pr}\left(0\right)=0.5} \\ {1,{\rm \; with\; Pr}\left(1\right)=0.5} \end{array}\right. , 
\end{equation} 
associated with the measurement of the $i$-th quantum system ${\left| g_{i}  \right\rangle} $ of ${\left| G \right\rangle} $ in \eqref{ZEqnNum755756}, and $M_{B} $ is a measurement in the computational basis $B$. 

Thus, the measurement of the $i$-th quantum system ${\left| g_{i}  \right\rangle} $ of ${\left| G \right\rangle} $ is defined via the following rule:
\begin{equation} \label{17)} 
b_{i} M_{B} =\left\{\begin{array}{l} {0,{\rm \; if\; }b_{i} =0,} \\ {M_{B} ,{\rm \; if\; }b_{i} =1} \end{array}\right. . 
\end{equation} 
In other words, the measurement result $M_{r} \left({\left| g_{i}  \right\rangle} \right)$ is kept only if $b_{i} =1$ in \eqref{ZEqnNum688653}; otherwise, the measurement result is discarded and replaced by a zero element. This results output $y_{i} $, as
\begin{equation} \label{ZEqnNum452944} 
y_{i} =\left\{\begin{array}{l} {0,{\rm \; if\; }b_{i} =0,} \\ {M_{B} \left({\left| g_{i}  \right\rangle} \right),{\rm \; if\; }b_{i} =1} \end{array}\right. . 
\end{equation} 
This measurement strategy defines $M_{r} $ \eqref{ZEqnNum688653} as a random Bernoulli vector \cite{ref20,ref21,ref22,ref23}. Then, the $n$-bit length output $Y$, is as
\begin{equation} \label{ZEqnNum634285} 
\begin{split}
   Y&={{M}_{r}}\left( \left| G \right\rangle  \right) \\ 
 & ={{M}_{r}}U( {\vec{{\theta }'}})\left| S \right\rangle  \\ 
 & ={{M}'_{r}}\left| S \right\rangle  \\ 
 & ={{\beta }_{C}}\Lambda  \\ 
 & ={{{{\beta }'}}_{C}}S,  
\end{split}
\end{equation} 
where $M'_{r} $ is
\begin{equation} \label{ZEqnNum152948} 
M'_{r} =M_{r} U(\vec{\theta '}) 
\end{equation} 
while $\beta _{C} $ is an $n$-length classical vector formulated via the $b_{i} $ bits of \eqref{ZEqnNum478723} as
\begin{equation} \label{ZEqnNum702773} 
\beta _{C} =\left(b_{1} ,\ldots ,b_{n} \right)^{T} ,                                                                
\end{equation} 
and
\begin{equation} \label{23)} 
\beta '_{C} =\beta _{C} U(\vec{\theta '}),                                                                   
\end{equation} 
and $\Lambda $ is 
\begin{equation} \label{24)} 
\Lambda =U(\vec{\theta '})S.                                                                        
\end{equation} 
As follows, applying $M_{r} $ \eqref{ZEqnNum688653} on ${\left| G \right\rangle} $ \eqref{ZEqnNum755756} is equivalent to applying $M'_{r} $ \eqref{ZEqnNum152948} on the computational basis state ${\left| S \right\rangle} $ \eqref{ZEqnNum536134}.

As \eqref{ZEqnNum634285} is determined via \eqref{ZEqnNum688653}, the goal is to determine $C\left(z\right)$ at a particular objective function $C$ via a post-processing ${\rm {\mathcal P}}$. 

First, from $Y$ \eqref{ZEqnNum634285}, the computational basis vector $S$ can be recovered as $\tilde{S}$ via ${\rm {\mathcal P}}$, as a minimization \cite{ref23}, 
\begin{equation} \label{ZEqnNum910643} 
\tilde{S}=\mathop{\arg \min }\limits_{S} L_{1} \left(S\right) 
\end{equation} 
such that
\begin{equation} \label{ZEqnNum128831} 
Y=\beta '_{C} S 
\end{equation} 
where $L_{1} $ is the $L_{1} $-norm. The ${\rm {\mathcal P}}$ unit utilizes a basis pursuit algorithm \cite{ref20,ref21,ref22,ref23} for the $L_{1} $-minimization in \eqref{ZEqnNum910643}. Then, using \eqref{ZEqnNum910643}, $\tilde{\Lambda }$ is defined as
\begin{equation} \label{ZEqnNum138005} 
\tilde{\Lambda }=U(\vec{\theta '})\tilde{S}.                                                                   
\end{equation} 
Thus, from \eqref{ZEqnNum138005}, the output vector $z$ is evaluated as
\begin{equation} \label{ZEqnNum667719} 
\begin{split}
   z&={{B}^{-1}}\mathcal{P}\left( Y \right) \\ 
 & ={{B}^{-1}}( {\tilde{\Lambda }}) \\ 
 & =U( {\vec{\theta }})\tilde{X},  
\end{split}
\end{equation} 
where ${\rm {\mathcal P}}\left(Y\right)$ is the post-processing \eqref{ZEqnNum910643} applied on $Y$, $B^{-1} $ is the inverse basis transformation and $\tilde{X}$ is a classical representation of ${\left| X \right\rangle} $. As follows, from \eqref{ZEqnNum667719}, the $C\left(z\right)$ estimate yields
\begin{equation} \label{29)} 
C\left(z\right)=C\left(B^{-1} {\rm {\mathcal P}}\left(Y\right)\right).                                                          
\end{equation} 
Then, assume that the procedure repeats for $R$ rounds. The $R$ rounds of dense measurement are defined via an $n\times R$ measurement matrix ${\rm {\mathcal M}}$ as
\begin{equation} \label{ZEqnNum508923} 
{\rm {\mathcal M}}=\left(M_{r}^{\left(1\right)} ,\ldots ,M_{r}^{\left(R\right)} \right),                                                             
\end{equation} 
where $M_{r}^{\left(m\right)} $ is an $n$-size random measurement vector \eqref{ZEqnNum688653} of the $m$-th measurement round $m=1,\ldots ,R$, as
\begin{equation} \label{ZEqnNum252234} 
M_{r}^{\left(m\right)} =\left(b_{1}^{\left(m\right)} M_{B} ,\ldots ,b_{n}^{\left(m\right)} M_{B} \right)^{T} ,                                                   
\end{equation} 
where $b_{i}^{\left(m\right)} $ is the $i$-th bit of $M_{r}^{\left(m\right)} $ defined via \eqref{ZEqnNum478723}, and ${M'_{r}}^{\left( m \right)}$ of the $m$-th round is 
\begin{equation} \label{32)} 
{M'_{r}}^{\left(m\right)} =M_{r}^{\left(m\right)} U(\vec{\theta '}),                                                             
\end{equation} 
and ${\beta '_{C}}^{\left( m \right)}$ of the $m$-th round is 
\begin{equation} \label{33)} 
{\beta '_{C}}^{\left(m\right)} =\beta _{C}^{\left(m\right)} U(\vec{\theta '}),                                                             
\end{equation} 
where 
\begin{equation} \label{ZEqnNum942401} 
\beta _{C}^{\left(m\right)} =\left(b_{1}^{\left(m\right)} ,\ldots ,b_{n}^{\left(m\right)} \right)^{T} . 
\end{equation} 
For the $R$ rounds, define the $n\times R$ orthogonal matrix ${\rm {\mathcal Q}}$ as
\begin{equation} \label{ZEqnNum527280} 
{\rm {\mathcal Q}}={\rm {\mathcal M}}U(\vec{\theta '})=\left({M'_{r}}^{\left(1\right)} ,\ldots ,{M'_{r}}^{\left(R\right)} \right),                                               
\end{equation} 
and the measurement output matrix $Y^{R} $ as
\begin{equation} \label{ZEqnNum395935} 
Y^{R} ={\rm {\mathcal Q}}{\left| S \right\rangle} =\left(Y^{\left(1\right)} ,\ldots ,Y^{\left(R\right)} \right),                                                    
\end{equation} 
where $Y^{\left(m\right)} $ is the measurement result vector \eqref{ZEqnNum128831} of the $m$-th round.

The problem is therefore to find the optimal value of $R$, such that the total error probability at the end of $R$ rounds
\begin{equation} \label{ZEqnNum977667} 
\Pr \left(z\ne z^{*} \right)=\xi  
\end{equation} 
picks up a given arbitrary value $\xi $ that is determined via the success of the $L_{1} $ minimization \eqref{ZEqnNum910643} in the ${\rm {\mathcal P}}$ unit.  

After some argumentations on the probability distribution of ${\rm {\mathcal Q}}$ \eqref{ZEqnNum527280}, at $R$ measurement rounds a concentration relation can be written as
\begin{equation} \label{38)} 
\begin{split}
  & \Pr \left( \left| {{\left( {{L}_{2}}\left( \mathcal{Q}\left| S \right\rangle  \right) \right)}^{2}}-{{\ell }^{2}}\left( \left| S \right\rangle  \right) \right|\ge \kappa \left( {{\ell }^{2}}\left( \left| S \right\rangle  \right) \right) \right) \\ 
 & =\Pr \left( \left| {{\left( {{L}_{2}}\left( \mathcal{M}\left| G \right\rangle  \right) \right)}^{2}}-{{\ell }^{2}}\left( \left| G \right\rangle  \right) \right|\ge \kappa \left( {{\ell }^{2}}\left( \left| G \right\rangle  \right) \right) \right) \\ 
 & \le 2\exp \left( -c{{\kappa }^{2}}R \right),  
\end{split}
\end{equation} 
where $c$ is a constant depending on the sub-Gaussian parameters $C_{1} ,C_{2} >0$ (see \sref{subg}) of the sub-Gaussian matrix ${\rm {\mathcal Q}}$ \eqref{ZEqnNum527280}, $L_{2} $ is the Euclidean norm and $\ell ^{2} $ is the $\ell ^{2} $-norm of a quantum system, $\ell ^{2} \left({\left| \psi  \right\rangle} \right)=\sqrt{\sum _{x}\left|\psi \left(x\right)\right|^{2}  } =1$, where $\left|\psi \left(x\right)\right|^{2} =\Pr \left(x\right)$ and  $\sqrt{\int \Pr \left(x\right)dx } =\sqrt{\int \left|\psi \left(x\right)\right|^{2} dx } =1$, while $\kappa $ is $\kappa \in \left(0,1\right)$. 

By theory, the $K$-th restricted isometry constant \cite{ref20,ref21,ref22,ref23} $\delta _{K} =\delta _{K} \left({\rm {\mathcal Q}}\right)$ of matrix ${\rm {\mathcal Q}}$ is the smallest $\chi \ge 0$ such that
\begin{equation} \label{ZEqnNum487772} 
\left(1-\chi \right)\ell ^{2} \left({\left| S \right\rangle} \right)\le \left(L_{2} \left({\rm {\mathcal Q}}{\left| S \right\rangle} \right)\right)^{2} \le \left(1+\chi \right)\ell ^{2} \left({\left| S \right\rangle} \right),                              
\end{equation} 
for $\forall S$ where $L_{0} \left(S\right)\le K$. 

Then, for a given $\chi $, the restricted isometry constant \cite{ref20,ref21,ref22,ref23} $\delta _{K} $ of ${\rm {\mathcal Q}}={\rm {\mathcal M}}U(\vec{\theta '})$ satisfies relation $\delta _{K} <\chi $ with probability 
\begin{equation} \label{40)} 
\Pr \left(\delta _{K} <\chi \right)=1-\varepsilon  
\end{equation} 
where $\varepsilon \in \left( 0,1 \right)$, if $R$ is selected as
\begin{equation} \label{ZEqnNum480800} 
R=A{\textstyle\frac{1}{\chi ^{2} }} \left(K\left(9+2\log \left({\textstyle\frac{n}{K}} \right)\right)+2\log \left(2\left({\textstyle\frac{1}{\varepsilon }} \right)\right)\right),                                       
\end{equation} 
where
\begin{equation} \label{42)} 
A={\textstyle\frac{2}{3c}} .                                                                     
\end{equation} 
The motivation for the selection of $R$ is as follows. The value of $R$ in \eqref{ZEqnNum480800} guarantees that the relation $\delta _{K} <\chi $ holds with probability $1-\varepsilon$, as it is given in \eqref{40)}. If $R$ is greater than \eqref{ZEqnNum480800}, then $\Pr \left(\delta _{K} <\chi \right)> 1-\varepsilon$, while if $R$ is lower than the value given in \eqref{ZEqnNum480800}, then $\Pr \left(\delta _{K} <\chi \right)< 1-\varepsilon$. As a corollary, the lowest value of $R$ to satisfy the relation $\delta _{K} <\chi $ with probability at least $1-\varepsilon$, is as given in \eqref{ZEqnNum480800}.
To prove \eqref{ZEqnNum480800}, express $\delta _{K} $ via \eqref{ZEqnNum487772} as
\begin{equation} \label{ZEqnNum763314} 
\delta _{K} =\sup _{\Upsilon \subset \left[n\right],\left|\Upsilon \right|=K} L_{2} \left({\rm {\mathcal Q}}_{\Upsilon }^{*} {\rm {\mathcal Q}}_{\Upsilon } -I\right),                                             
\end{equation} 
where $\Upsilon $ is subset, ${\rm {\mathcal Q}}_{\Upsilon } $ is a submatrix, $I$ is the identity matrix, $\left|\Upsilon \right|$ is the cardinality of subset $\Upsilon $ and $\left[n\right]=\left\{1,\ldots ,n\right\}$ is the set of natural numbers not exceeding $n$. 

The formula of \eqref{ZEqnNum763314} is equivalent to \eqref{ZEqnNum487772}, since \eqref{ZEqnNum487772} can be rewritten as
\begin{equation} \label{44)} 
\left|\left(L_{2} \left({\rm {\mathcal Q}}_{\Upsilon } {\left| S' \right\rangle} \right)\right)^{2} -\ell ^{2} \left({\left| S' \right\rangle} \right)\right|\le \chi \ell ^{2} \left({\left| S' \right\rangle} \right) 
\end{equation} 
for $\forall \Upsilon \subset \left[n\right]$, $\left|\Upsilon \right|\le K$ and $S'\subset S$. Let 
\begin{equation} \label{45)} 
Y_{\Upsilon } ={\rm {\mathcal Q}}_{\Upsilon } {\left| S' \right\rangle} ,                                                                   
\end{equation} 
and
\begin{equation} \label{46)} 
Z_{\Upsilon } =H{\left| S' \right\rangle} ,                                                                   
\end{equation} 
where $H$ is a Hermitian matrix, 
\begin{equation} \label{ZEqnNum108272} 
H={\rm {\mathcal Q}}_{\Upsilon }^{*} {\rm {\mathcal Q}}_{\Upsilon } -I,                                                                  
\end{equation} 
then
\begin{equation} \label{48)} 
\begin{split}
  & {{\left( {{L}_{2}}\left( {{Y}_{\Upsilon }} \right) \right)}^{2}}-{{\ell }^{2}}\left( \left| {{S}'} \right\rangle  \right) \\ 
 & =\left\langle {{Y}_{\Upsilon }},{{Y}_{\Upsilon }} \right\rangle -\left\langle {S}',{S}' \right\rangle  \\ 
 & =\left\langle {{Z}_{\Upsilon }},{S}' \right\rangle   
\end{split}
\end{equation} 
Therefore, $L_{2} \left(H\right)$ can be expressed as a maximization
\begin{equation} \label{49)} 
L_{2} \left(H\right)=\mathop{\max }\limits_{S'} {\textstyle\frac{\left\langle Z_{\Upsilon } ,S'\right\rangle }{\ell ^{2} \left({\left| S' \right\rangle} \right)}}  
\end{equation} 
that leads to relation
\begin{equation} \label{50)} 
\mathop{\max }\limits_{\Upsilon \subset \left[n\right],\left|\Upsilon \right|=K} L_{2} \left(H\right)\le \chi .                                                        
\end{equation} 
Then, the union bound takes over all $(\begin{smallmatrix}
   n  \\
   K  \\
\end{smallmatrix})$ subsets $\Upsilon \subset \left[n\right]$ of cardinality $K$, yields the relation of
\begin{equation} \label{ZEqnNum616872} 
\begin{split}
   \Pr \left( {{\delta }_{K}}\ge \chi  \right)& \le \sum\limits_{{{\sup }_{\Upsilon \subset \left[ n \right],\left| \Upsilon  \right|=K}}}{\Pr \left( {{L}_{2}}\left( H \right)\ge \chi  \right)} \\ 
 & \le 2\left( \begin{smallmatrix}
   n  \\
   K  \\
\end{smallmatrix} \right){{\left( 1+\tfrac{2}{\Omega } \right)}^{K}}\exp \left( -c{{\chi }^{2}}{{\left( 1-2\Omega  \right)}^{2}}R \right) \\ 
 & \le 2{{\left( \tfrac{en}{K} \right)}^{K}}{{\left( 1+\tfrac{2}{\Omega } \right)}^{K}}\exp \left( -c{{\chi }^{2}}{{\left( 1-2\Omega  \right)}^{2}}R \right),  
\end{split}
\end{equation} 
where we used that for integers $m\ge k>0$,  $\left({\textstyle\frac{m}{k}} \right)^{k} \le (\begin{smallmatrix}
   m  \\
   k  \\
\end{smallmatrix}) \le \left({\textstyle\frac{em}{k}} \right)^{k},$ by theory \cite{ref20,ref21,ref22,ref23}.

It can be verified that in \eqref{ZEqnNum616872} for $\Upsilon \subset \left[n\right]$ with $\left|\Upsilon \right|=K$, the relation
\begin{equation} \label{ZEqnNum849837} 
\Pr \left(L_{2} \left(H\right)<\chi \right)=1-\varepsilon  
\end{equation} 
holds, if  
\begin{equation} \label{ZEqnNum669620} 
R={\textstyle\frac{2}{3c\chi ^{2} }} \left(7K+2\log \left(2\left({\textstyle\frac{1}{\varepsilon }} \right)\right)\right), 
\end{equation} 
since for  
\begin{equation} \label{ZEqnNum678806} 
\kappa =\left(1-2\Omega \right)\chi  
\end{equation} 
it can be verified that
\begin{equation} \label{ZEqnNum347354} 
\Pr \left(L_{2} \left(H\right)\ge \chi \right)\le 2\left(1+{\textstyle\frac{2}{\Omega }} \right)^{K} \exp \left(-c\left(1-2\Omega \right)^{2} \chi ^{2} R\right).                         
\end{equation} 
Thus, \eqref{ZEqnNum849837} is satisfied only if 
\begin{equation} \label{ZEqnNum582908} 
R={\textstyle\frac{1}{c\left(1-2\Omega \right)^{2} \chi ^{2} }} \left(\log \left(1+{\textstyle\frac{2}{\Omega }} \right)K+\log \left(2\left({\textstyle\frac{1}{\varepsilon }} \right)\right)\right). 
\end{equation} 
Then, setting $\Omega $ in \eqref{ZEqnNum582908} to
\begin{equation} \label{ZEqnNum288645} 
\Omega ={\textstyle\frac{2}{e^{3.5} -1}}  
\end{equation} 
so that
\begin{equation} \label{58)} 
{\textstyle\frac{1}{\left(1-2\Omega \right)^{2} }} \le {\textstyle\frac{4}{3}}  
\end{equation} 
and
\begin{equation} \label{59)} 
\log \left(1+{\textstyle\frac{2}{\Omega }} \right){\textstyle\frac{1}{\left(1-2\Omega \right)^{2} }} \le {\textstyle\frac{14}{3}} ,                                                          
\end{equation} 
yields \eqref{ZEqnNum669620} \cite{ref23}.

Note that it also can be shown that for $\Omega \in \left(0,0.5\right)$ in \eqref{ZEqnNum288645}, there exists a finite subset $\Gamma $ of a unit ball ${\rm {\mathcal B}}_{\Upsilon } =\left\{X,{\rm supp}X\subset \Upsilon ,\ell ^{2} \left({\left| X \right\rangle} \right)\le 1\right\}$ such that $\left|\Gamma \right|$ is 
\begin{equation} \label{60)} 
\left|\Gamma \right|\le \left(1+{\textstyle\frac{2}{\Omega }} \right)^{K}  
\end{equation} 
and
\begin{equation} \label{61)} 
\mathop{\min }\limits_{x\in \Gamma } L_{2} \left(z-x\right)\le \Omega  
\end{equation} 
for $\forall z\in {\rm {\mathcal B}}_{\Upsilon } $, such that for $x\subset \Gamma $
\begin{equation} \label{62)} 
\begin{split}
   \Pr& \left( \left| {{\left( {{L}_{2}}\left( \mathcal{Q}\left| x \right\rangle  \right) \right)}^{2}}-{{\ell }^{2}}\left( \left| x \right\rangle  \right) \right|\ge \kappa \left( {{\ell }^{2}}\left( \left| x \right\rangle  \right) \right) \right) \\ 
 & \le \sum\limits_{x\in \Gamma }{\Pr \left( \left| {{\left( {{L}_{2}}\left( \mathcal{Q}\left| x \right\rangle  \right) \right)}^{2}}-{{\ell }^{2}}\left( \left| x \right\rangle  \right) \right|\ge \kappa \left( {{\ell }^{2}}\left( \left| x \right\rangle  \right) \right) \right)} \\ 
 & \le 2\left| \Gamma  \right|\exp \left( -c{{\kappa }^{2}}R \right) \\ 
 & \le 2{{\left( 1+\tfrac{2}{\Omega } \right)}^{K}}\exp \left( -c{{\kappa }^{2}}R \right),  
\end{split}
\end{equation} 
and 
\begin{equation} \label{ZEqnNum983488} 
\begin{split}
   \Pr &\left( \left| {{\left( {{L}_{2}}\left( \mathcal{Q}\left| x \right\rangle  \right) \right)}^{2}}-{{\ell }^{2}}\left( \left| x \right\rangle  \right) \right|<\kappa \left( {{\ell }^{2}}\left( \left| x \right\rangle  \right) \right),\forall x\subset \Gamma  \right) \\ 
 & =1-2{{\left( 1+\tfrac{2}{\Omega } \right)}^{S}}\exp \left( -c{{\kappa }^{2}}R \right).  
\end{split}
\end{equation} 
By finding the values of $\Omega $ and $\kappa $, the relation 
\begin{equation} \label{64)} 
\left|\left(L_{2} \left({\rm {\mathcal Q}}{\left| z \right\rangle} \right)\right)^{2} -\ell ^{2} \left({\left| z \right\rangle} \right)\right|=L_{2} \left(H\right)\le \chi  
\end{equation} 
can be satisfied for $\forall z\in {\rm {\mathcal B}}_{\Upsilon } $. 

It can be proven at $H$ \eqref{ZEqnNum108272} and
\begin{equation} \label{65)} 
W=H{\left| x \right\rangle} ,                                                                       
\end{equation} 
for $\forall x\subset \Gamma $ that the relation
\begin{equation} \label{66)} 
\left|\left\langle W,x\right\rangle \right|<\kappa ,                                                                     
\end{equation} 
holds. Thus, for a given $z$ and $x\subset \Gamma $, such that $L_{2} \left(z-x\right)\le \Omega \le {\textstyle\frac{1}{2}} $, 
\begin{equation} \label{67)} 
\begin{split}
   \left| \left\langle V,z \right\rangle  \right|&=\left| \left\langle W,x \right\rangle +\left\langle D,z-x \right\rangle  \right| \\ 
 & \le \left| \left\langle W,x \right\rangle  \right|+\left| \left\langle D,z-x \right\rangle  \right| \\ 
 & <\kappa +{{L}_{2}}\left( H \right)\sqrt{{{\ell }^{2}}\left( \left| z+x \right\rangle  \right)}\sqrt{{{\ell }^{2}}\left( \left| z-x \right\rangle  \right)} \\ 
 & \le \kappa +2{{L}_{2}}\left( H \right)\Omega   
\end{split}
\end{equation} 
where 
\begin{equation} \label{68)} 
V=H{\left| z \right\rangle} ,                                                                      
\end{equation} 
and 
\begin{equation} \label{69)} 
D=H{\left| z+x \right\rangle} .                                                                  
\end{equation} 
Then, a maximization over $\forall z\in {\rm {\mathcal B}}_{\Upsilon } $ yields
\begin{equation} \label{70)} 
L_{2} \left(H\right)<\kappa +2L_{2} \left(H\right)\Omega .                                                           
\end{equation} 
Thus,
\begin{equation} \label{71)} 
L_{2} \left(H\right)\le {\textstyle\frac{\kappa }{1-2\Omega }} .                                                                    
\end{equation} 
As follows, there exists \eqref{ZEqnNum678806} such that $L_{2} \left(H\right)<\chi $ holds, and combining it with \eqref{ZEqnNum983488} verifies the relation of \eqref{ZEqnNum347354}.

To conclude the results, setting $\Omega $ in \eqref{ZEqnNum616872} with equality in \eqref{ZEqnNum288645} leads to $\delta _{K} <\chi $ with probability $\Pr \left(\delta _{K} <\chi \right)=1-\varepsilon $, as the $R$ value of measurement rounds is
\begin{equation} \label{ZEqnNum690736} 
\begin{split}
   R&=\tfrac{1}{c{{\chi }^{2}}}\left( \tfrac{4}{3}K\log \left( \tfrac{10n}{K} \right)+\tfrac{14}{3}K+\tfrac{4}{3}\log \left( 2\left( \tfrac{1}{\varepsilon } \right) \right) \right) \\ 
 & =\tfrac{2}{3c}\tfrac{1}{{{\chi }^{2}}}\left( K\left( 9+2\log \left( \tfrac{n}{K} \right) \right)+2\log \left( 2\left( \tfrac{1}{\varepsilon } \right) \right) \right).  
\end{split}
\end{equation} 

\end{proof}

Note that if $R$ is selected to be greater than \eqref{ZEqnNum690736}, the probability is increased to $\Pr \left(\delta _{K} <\chi \right)>1-\varepsilon $.

\subsection{Dense Measurement Rounds in Gate-Model Quantum Computers}
\subsubsection{Arbitrary Unitary Sequences}
The next theorem reveals that the number $R$ of dense measurement rounds can be used to determine $z^{*} $ with an error probability $\zeta =n^{-\log ^{3} \left(n\right)} $, such that $R$ depends only on the properties of the unitaries, while it does not depend directly on the actual $\zeta $. 
\begin{theorem}
(Dense measurements at a $U(\vec{\theta })$ quantum gate structure). For an arbitrary unitary $U_{B} $ in $U(\vec{\theta })=U_{B} U(\vec{\theta '})$ with $L_{0} \left(S\right)\le K$, the global optimal $z^{*} $ and estimate $C\left(z^{*} \right)$ can be determined via $R=\alpha Z^{2} K\log ^{4} \left(n\right)$ dense measurement rounds, with probability $\Pr \left(z^{*} \right)=\Pr \left(C\left(z^{*} \right)\right)=1-n^{-\log ^{3} \left(n\right)} $, where $Z\ge \sqrt{n} \mathop{\max }\limits_{k,q\in \left[n\right]} \left|U(\vec{\theta }_{q,k})\right|$, $U(\vec{\theta }_{q,k})$ is the $q$-th element of the $k$-th column of $U(\vec{\theta })$, while $\alpha >0$ is a constant. 
\end{theorem}
\begin{proof}
Let assume that $U(\vec{\theta })$ can be decomposed as $U_{B} U(\vec{\theta '})$, and the following bound can be formulated for the entries of $U(\vec{\theta })$, 
\begin{equation} \label{ZEqnNum109677} 
\mathop{\max }\limits_{k,q\in \left[n\right]} \left|U(\vec{\theta }_{q,k})\right|\le {\textstyle\frac{Z}{\sqrt{n} }} ,                                                             
\end{equation} 
where $U(\vec{\theta }_{q,k})$ is the $q$-th element of the $k$-th column of $U(\vec{\theta })$, and
\begin{equation} \label{74)} 
\left|U(\vec{\theta }_{q,k})\right|=\sqrt{U(\vec{\theta }_{q,k})(U(\vec{\theta }_{q,k}))^{*} } .                                                
\end{equation} 
Let assume that the size of $U(\vec{\theta })$ is $n\times n$, with columns $u_{k} $, $k=1,\ldots ,n$. Then let $v_{k} $ be the normalization of $u_{k} $ as
\begin{equation} \label{ZEqnNum586932} 
v_{k} =\sqrt{n} u_{k} ,                                                                        
\end{equation} 
where the normalized columns form an orthonormal system, and let $\varphi _{kl} $ be the inner product of two normalized columns $v_{k} $ and $v_{l} $, as
\begin{equation} \label{ZEqnNum934641} 
\varphi _{kl} =\left\langle {\textstyle\frac{1}{\sqrt{n} }} v_{k} ,{\textstyle\frac{1}{\sqrt{n} }} v_{l} \right\rangle =\left\langle u_{k} ,u_{l} \right\rangle ,                                                 
\end{equation} 
that can be rewritten as
\begin{equation} \label{77)} 
\varphi _{kl} ={\textstyle\frac{1}{n}} \sum _{q=1}^{n}\sqrt{n} U(\vec{\theta }_{q,k}) \sqrt{n} U^{\dag } (\vec{\theta }_{q,l})={\textstyle\frac{1}{n}} \sum _{q=1}^{n}\sqrt{n} u_{k,q}  \sqrt{n} u_{l,q}^{\dag } , 
\end{equation} 
where $u_{i,j} =U(\vec{\theta }_{j,i})$ and $v_{i,j} =\sqrt{n} U(\vec{\theta }_{j,i})$. Therefore, in \eqref{77)}, the sum operator runs over the $n$ elements of the $k$-th column of unitary $U(\vec{\theta })$, and the $n$ elements of the $l$-th column of $U^{\dag } (\vec{\theta })$, respectively.


Then, at $U_{B} $ and $U(\vec{\theta '})$, some argumentations on bounded orthonormal systems straightforwardly yields the boundedness condition \cite{ref23}
\begin{equation} \label{ZEqnNum139322} 
Z\ge \mathop{\max }\limits_{k,q\in \left[n\right]} \left|\left\langle b_{q} ,u'_{k} \right\rangle \right|, 
\end{equation} 
where $b_{q} $ is the $q$-th column of  $U_{B} $.

Then, for the maximal entry of $U(\vec{\theta }_{q,k})$, a bound can be established via the normalized columns, as
\begin{equation} \label{ZEqnNum493727} 
\begin{split}
   Z \ge \underset{k,q\in \left[ n \right]}{\mathop{\max }}\,\left| {{v}_{k,q}} \right|&=\underset{k,q\in \left[ n \right]}{\mathop{\max }}\,\left| \sqrt{n}{{u}_{k,q}} \right| \\ 
 & =\sqrt{n}\underset{k,q\in \left[ n \right]}{\mathop{\max }}\,\left| {{u}_{k,q}} \right|=\sqrt{n}\underset{k,q\in \left[ n \right]}{\mathop{\max }}\,\left| U( {{{\vec{\theta }}}_{q,k}}) \right|.  
\end{split}
\end{equation} 
As follows, the bounds in \eqref{ZEqnNum139322} and \eqref{ZEqnNum493727} are equivalent to \eqref{ZEqnNum109677}.

Then, by introducing a projector ${\rm {\mathcal P}}_{Q_{R} } $ that selects a subset of $U(\vec{\theta })$ in the $R$ rounds, the ${\rm {\mathcal M}}$ \eqref{ZEqnNum508923} measurement operator applied on a unitary $U(\vec{\theta })$ can be rewritten as 
\begin{equation} \label{ZEqnNum432717} 
{\rm {\mathcal M}}={\rm {\mathcal P}}_{Q_{R} } \left(U(\vec{\theta })\right),                                                               
\end{equation} 
where $Q_{R} \subset \left[n\right]$ is a subset of $R$ elements selected uniform at random from all subsets of $\left[n\right]$ of cardinality $R$, $\left|Q_{R} \right|=R$,
\begin{equation} \label{ZEqnNum869135} 
Q_{R} =\left\{q_{1} ,\ldots ,q_{R} \right\}.                                                               
\end{equation} 
As follows, the $Y^{R} $ \eqref{ZEqnNum395935} measurement result can be rewritten as
\begin{equation} \label{ZEqnNum330472} 
Y^{R} ={\rm {\mathcal P}}_{Q_{R} } \left(U(\vec{\theta })\right){\left| X \right\rangle} ={\rm {\mathcal P}}_{Q_{R} } \left(U(\vec{\theta '})U_{B} \right){\left| X \right\rangle} ={\rm {\mathcal P}}_{Q_{R} } \left(U(\vec{\theta '})\right){\left| S \right\rangle} .                     
\end{equation} 
It is required to verify that the $\xi $ error probability \eqref{ZEqnNum977667} at a projector ${\rm {\mathcal P}}_{Q} $ in \eqref{ZEqnNum432717} is bounded by an $\xi ^{*} $ error probability associated with the selection of rows uniformly and independently at random from $U(\vec{\theta })$ \cite{ref23}.

Thus, we define set $Q'_{R} \subset \left[n\right]$ with the same cardinality as \eqref{ZEqnNum869135}, such that its elements are selected independently and uniformly at random from $\left[n\right]$, $\left|Q'_{R} \right|=R$ 
\begin{equation} \label{84)} 
Q'_{R} =\left\{q'_{1} ,\ldots ,q'_{R} \right\}.                                                             
\end{equation} 
Then, let $Q_{k} \subset \left[n\right]$ be a subset of $k\le R$ selected uniform at random from all subsets of $\left[n\right]$ of cardinality $k$, $\left|Q_{k} \right|=k$,
\begin{equation} \label{85)} 
Q_{k} =\left\{q_{1} ,\ldots ,q_{k} \right\}.                                                              
\end{equation} 
For any subset $Q\in \left[n\right]$, we define a failure event ${\rm {\mathcal E}}\left(Q\right)$ as  
\begin{equation} \label{86)} 
{\rm {\mathcal E}}\left(Q\right)\equiv \left\{\tilde{S}\ne \mathop{\arg \min }\limits_{S} L_{1} \left(S\right),{\rm \; s.t.\; }Y={\rm {\mathcal P}}_{Q} \left(U(\vec{\theta '})\right){\left| S \right\rangle} ,{\rm \; for\; }\forall {\left| S \right\rangle} \right\},                
\end{equation} 
i.e., the event that the $L_{1} $-minimization (i.e,. a basis pursuit algorithm in ${\rm {\mathcal P}}$) allows no to determine every $S$ from \eqref{ZEqnNum330472} on $Q$ (Note that the success probability of an $L_{1} $-minimization in ${\rm {\mathcal P}}$ to determine $S$ is independent from the normalization of the measurement operator.). 

It can be verified, that for $Q\subset \tilde{Q}$, 
\begin{equation} \label{87)} 
{\rm {\mathcal E}}(\tilde{Q})\subset {\rm {\mathcal E}}\left(Q\right),                                                                 
\end{equation} 
and for $k\le R$, 
\begin{equation} \label{88)} 
\Pr \left({\rm {\mathcal E}}\left(Q_{R} \right)\right)=\xi \le \Pr \left({\rm {\mathcal E}}\left(Q_{k} \right)\right)=\xi _{k} ,                                                
\end{equation} 
and if $\left|Q'_{R} \right|=k$ holds for $k\le R$, then
\begin{equation} \label{89)} 
{\rm {\mathcal D}}\left(Q'_{R} \right)={\rm {\mathcal D}}\left(Q_{k} \right),                                                                
\end{equation} 
where ${\rm {\mathcal D}}\left(\cdot \right)$ is the distribution.  

Therefore, the 
\begin{equation} \label{90)} 
\Pr \left({\rm {\mathcal E}}\left(Q'_{R} \right)\right)=\xi ^{*}  
\end{equation} 
probability of event ${\rm {\mathcal E}}\left(Q'_{R} \right)$ at $\left|Q'_{R} \right|=k$ is as
\begin{equation} \label{91)} 
\begin{split}
   {{\xi }^{*}}&=\sum\limits_{k=1}^{R}{\Pr \left( \left. \mathcal{E}\left( {{{{Q'_{R}}}}} \right) \right|\left( \left| {{{{Q'_{R}}}}} \right|=k \right) \right)\Pr \left( \left| {{{{Q'_{R}}}}} \right|=k \right)} \\ 
 & =\sum\limits_{k=1}^{R}{{{\xi }_{k}}\Pr \left( \left| {{{{Q'_{R}}}}} \right|=k \right)\ge \xi }\sum\limits_{k=1}^{R}{\Pr \left( \left| {{{{Q'_{R}}}}} \right|=k \right)} \\ 
 & =\xi ,  
\end{split}
\end{equation} 
thus the error probability $\xi $ is bounded by $\xi ^{*} $. 

As follows, using projector ${\rm {\mathcal P}}_{Q} $ in \eqref{ZEqnNum432717}, an $L_{1} $-minimization (basis pursuit) \cite{ref20,ref21,ref22,ref23} in the ${\rm {\mathcal P}}$ post-processing phase allows to determine the global optimal $z^{*} $ from \eqref{ZEqnNum330472}, with probability 
\begin{equation} \label{ZEqnNum219291} 
\Pr \left(z^{*} \right)=1-n^{-\log ^{3} \left(n\right)},
\end{equation} 
as 
\begin{equation} \label{ZEqnNum602021} 
R=\alpha Z^{2} K\log ^{4} \left(n\right) 
\end{equation} 
holds, where $\alpha >0$, and $\xi $ is evaluated via \eqref{ZEqnNum977667} as 
\begin{equation} \label{ZEqnNum846366} 
\xi =\Pr \left(z\ne z^{*} \right)=n^{-\log ^{3} \left(n\right)} .                                               
\end{equation} 

The value determined for $R$ in \eqref{ZEqnNum602021} is based on the following fact. It can be shown \cite{ref23}, that at a particular $\alpha >0$ and $Z \ge 1$, there exists a constant $a$, $a\in \left( 0,1 \right)$, such that for a given $a$, the relation 
\begin{equation}
\delta _{K} \le a 
\end{equation}
holds with probability 
\eqref{ZEqnNum219291}, if
\begin{equation} \label{add2} 
R\ge \tfrac{1}{{{a}^{2}}}\alpha {{Z}^{2}}K{{\log }^{4}}\left( n \right),
\end{equation} 
where 
\begin{equation} 
\delta _{K} =\delta _{K} \left({\tfrac{1}{\sqrt{R}}}{\rm {\mathcal Q}}\right)
\end{equation} 
 is the $K$-th restricted isometry constant of ${\tfrac{1}{\sqrt{R}}}{\rm {\mathcal Q}}$, that is the smallest $a \ge 0$ such that 
\begin{equation} \label{add3)} 
\left(1-a\right)\ell ^{2} \left({\left| S \right\rangle} \right)\le \left(L_{2} \left({\tfrac{1}{\sqrt{R}}}{\rm {\mathcal Q}}){\left| S \right\rangle} \right)\right)^{2} \le \left(1+a\right)\ell ^{2} \left({\left| S \right\rangle} \right)
\end{equation} 
while ${\rm {\mathcal Q}}$ is as given by \eqref{ZEqnNum527280}; holds for $\forall S$, with $L_{0} \left(S\right)\le K$. Since for \eqref{ZEqnNum602021}, $\alpha >0$ holds and $Z \ge 1$ is satisfied via \eqref{ZEqnNum493727}, the result in \eqref{ZEqnNum602021} is straightforwardly follows from in \eqref{add2} at
\begin{equation} 
a=1. 
\end{equation}

Thus, at \eqref{ZEqnNum602021} dense measurement rounds the global optimal output $z^{*} $
\begin{equation} \label{ZEqnNum735912} 
z^{*} =B^{-1} \left(\Lambda ^{*} \right),                                                                   
\end{equation} 
is yielded with probability 
\begin{equation} \label{96)} 
\Pr \left(z^{*} \right)=1-\xi ,                                                                  
\end{equation} 
where $\Lambda ^{*} $ is the optimal $\tilde{\Lambda }$ determined via ${\rm {\mathcal P}}$. Therefore the $C\left(z^{*} \right)$ global optimal estimate of a particular objective function $C$ can be determined with probability $\Pr \left(C\left(z^{*} \right)\right)=\Pr \left(z^{*} \right)$ as
\begin{equation} \label{ZEqnNum742810} 
C\left(z^{*} \right)=\mathop{\max }\limits_{\forall m} C\left(z_{m} \right),                                                                  
\end{equation} 
where $C\left(z_{m} \right)$ is the estimate yielded in an $m$-th round,  $m=1,\ldots ,R$, $z_{m} $ is the recovered output vector in the $m$-th round, which concludes the proof. 
\end{proof}

The steps of the dense measurement for an arbitrary $U(\vec{\theta })$ are summarized in Procedure 1.

 \setcounter{algocf}{0}
\begin{proced}
  \DontPrintSemicolon
\caption{Dense Measurements at arbitrary unitary $U(\vec{\theta })$}

\textbf{Step 1}. Set the superposed input system ${\left| X \right\rangle} $ \eqref{ZEqnNum910801} and unitary sequence $U(\vec{\theta })$ \eqref{ZEqnNum489626} of $QG$.

\textbf{Step 2}. Select a computational basis $B$ for unitary $U_{B} $ to set ${\left| S \right\rangle} =U_{B} {\left| X \right\rangle} $, such that $L_{0} \left(S\right)\le K$ holds for the $L_{0} $-norm of $S=BX$.

\textbf{Step 3}. At a unitary $U_{B} $, set the gate parameter vector $\vec{\theta '}$ in $U(\vec{\theta '})$, such that $U(\vec{\theta })=U_{B} U(\vec{\theta '})$.

\textbf{Step 4}. Measure ${\left| G \right\rangle} =U(\vec{\theta '}){\left| S \right\rangle} $ via $M_{r} $ to get $Y=M_{r} \left(U(\vec{\theta '}){\left| S \right\rangle} \right)$.

\textbf{Step 5}. Apply ${\rm {\mathcal P}}$ post-processing to determine $\tilde{S}$ via \eqref{ZEqnNum910643} and $\tilde{\Lambda }=U(\vec{\theta '})\tilde{S}$ in \eqref{ZEqnNum138005}. 

\textbf{Step 6}. Output $z$ as in \eqref{ZEqnNum667719}, where $B^{-1} $ is the inverse of $B$.

\textbf{Step 7}. Apply steps 1-6 through $R$ rounds \eqref{ZEqnNum602021}, to achieve error probability $\xi $ \eqref{ZEqnNum846366} in the determination of $z^{*} $ \eqref{ZEqnNum735912} and the global optimal objective function $C\left(z^{*} \right)$ \eqref{ZEqnNum742810}.
\end{proced} 

\fref{fig1} depicts an application of dense quantum measurement in quantum computations. \fref{fig1}(a) shows an arbitrary $QG$ quantum circuit with a $U(\vec{\theta })$ sequence (where $\vec{\theta }$ sets the unitaries of the quantum circuit) and standard measurement $M$ with $R_{0} $ measurement rounds and success probability $\Pr _{R_{0} } \left(z^{*} \right)$. \fref{fig1}(b) shows dense quantum measurement in a general case $U(\vec{\theta })$, such that the $QG$ structure is prepared to realize $U(\vec{\theta })=U_{B} U(\vec{\theta '})$, where unitary $U_{B} $ sets a computational basis as ${\left| S \right\rangle} =U_{B} {\left| X \right\rangle} $ with relation $L_{0} \left(S\right)$ for the $L_{0} $-norm of $S$ ($S$ is a classical value resulting from the measurement), while $U(\vec{\theta '})$ is the actual setting of the unitaries of $QG$ to provide output ${\left| G \right\rangle} =U(\vec{\theta '}){\left| S \right\rangle} $, such that $U(\vec{\theta '}){\left| S \right\rangle} =U(\vec{\theta }){\left| X \right\rangle} $. The output is measured via a randomized measurement $M_{r} $. The measurement result is then post-processed via unit ${\rm {\mathcal P}}$ to achieve $R{\rm \ll }R_{0} $ and $\Pr \left(z^{*} \right){\rm \gg }\Pr _{R_{0} } \left(z^{*} \right)$.

\begin{center}
\begin{figure*}[!htbp]
\begin{center}
\includegraphics[angle = 0,width=1\linewidth]{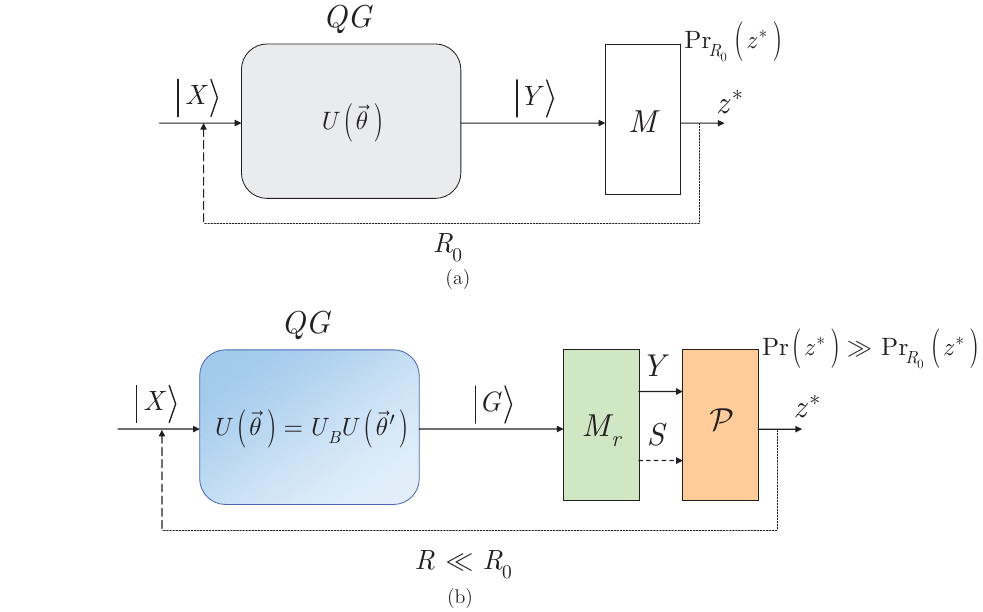}
\caption{(a). A standard quantum measurement setting. The ${\left| X \right\rangle} $ superposed input quantum system is fed into the $QG$ quantum gate structure. The $QG$ structure realizes the unitary $U(\vec{\theta })$ and outputs ${\left| Y \right\rangle} =U(\vec{\theta }){\left| X \right\rangle} $. The output is measured by a standard (projective) measurement operator $M$ in a computational basis, which results in an optimal output $z^{*} $ after $R_{0} $ measurement rounds (depicted by the dashed line) with a success probability of $\Pr _{R_{0} } \left(z^{*} \right)$. (b): The dense quantum measurement procedure. The ${\left| X \right\rangle} $ superposed input quantum system is fed into the $QG$ quantum gate structure. The $QG$ structure is set to realize the unitary $U(\vec{\theta })=U_{B} U(\vec{\theta '})$. The $M_{r} $ randomized measurement yields $Y=M_{r} \left(U(\vec{\theta '}){\left| S \right\rangle} \right)$, which is post-processed along with $S=BX$ via unit ${\rm {\mathcal P}}$. Unit ${\rm {\mathcal P}}$ performs an $L_{1} $-norm minimization and outputs the global optimal $z^{*} $ after $R{\rm \ll }R_{0} $ rounds (depicted by the dashed line), with a high success probability of $\Pr \left(z^{*} \right){\rm \gg }\Pr _{R_{0} } \left(z^{*} \right)$.} 
 \label{fig1}
 \end{center}
\end{figure*}
\end{center}

\subsubsection{Dense Measurements of Computational Basis Quantum States}
The next theorem reveals that the number of dense measurement rounds can be reduced if the unitaries of the quantum circuit are set as $U(\vec{\theta })=U_{B} $, i.e., if the output of the quantum circuit is a computational basis state ${\left| S \right\rangle} =U_{B} {\left| X \right\rangle} $.
\begin{theorem}
(Number of dense measurement rounds at $U(\vec{\theta })=U_{B} $). At $U(\vec{\theta })=U_{B} $, the optimal $z^{*} $ and $C\left(z^{*} \right)$ can be determined via $R\left( \xi  \right)={{c}_{1}}K\log \left( {10n}/{K}\right)+{{c}_{2}}\log \left( {2}/{\xi }\right)$ dense measurement rounds, where $\xi \in \left(0,1\right)$ is the error probability of $z^{*} $, while $c_{1} >0$ and $c_{2} >0$ are constants, that yields $R={\rm {\mathcal O}}\left(K\log \left({\textstyle\frac{10n}{K}} \right)\right)$, with $\xi =2\exp \left(-R\right)$.
\end{theorem}
\begin{proof}
In this setting, the unitaries of the $QG$ quantum gate structure are set such that 
\begin{equation} \label{ZEqnNum171293} 
U(\vec{\theta })=U_{B}  
\end{equation} 
holds, therefore the output of $QG$ at an $n$-length input ${\left| X \right\rangle} $ is 
\begin{equation} \label{99)} 
U_{B} {\left| X \right\rangle} ={\left| S \right\rangle} ,                                                                 
\end{equation} 
i.e., it outputs an $n$-length computational basis quantum state ${\left| S \right\rangle} $, with relation $L_{0} \left(S\right)\le K$, $K{\rm \ll }n$.

The $R$ measurements are performed according to the $n\times R$ measurement matrix ${\rm {\mathcal M}}$ as defined in \eqref{ZEqnNum508923}.

Since $U(\vec{\theta })$ is set as given in \eqref{ZEqnNum171293}, the measurement results of the $R$ rounds formulate $R$ dimensional output $Y^{R} =\left(Y^{\left(1\right)} ,\ldots ,Y^{\left(R\right)} \right)$ as
\begin{equation} \label{100)} 
Y^{R} ={\rm {\mathcal M}}{\left| S \right\rangle} ,                                                                  
\end{equation} 
where the output of the $m$-th round is 
\begin{equation} \label{ZEqnNum485240} 
Y^{\left(m\right)} =M_{r}^{\left(m\right)} {\left| S \right\rangle} =\beta _{C}^{\left(m\right)} S, 
\end{equation} 
where $M_{r}^{\left(m\right)} $ is as given in \eqref{ZEqnNum252234}, while $\beta _{C}^{\left(m\right)} $ is as in \eqref{ZEqnNum942401}.

It further can be verified that ${\rm {\mathcal M}}$ is a sub-Gaussian random matrix, thus there exists a constant $\lambda >0$ depending only on the $C_{1} ,C_{2} $ sub-Gaussian parameters (see \sref{subg}) of ${\rm {\mathcal M}}$ such that for a given $\chi $, the restricted isometry constant \cite{ref20,ref21,ref22,ref23} $\delta _{K} =\delta _{K} \left({\rm {\mathcal M}}\right)$ of ${\rm {\mathcal M}}$, satisfies relation $\delta _{K} <\chi $ with probability 
\begin{equation} \label{ZEqnNum134847} 
\Pr \left(\delta _{K} <\chi \right)=1-\varepsilon ,                                                             
\end{equation} 
if 
\begin{equation} \label{ZEqnNum898189} 
R=\lambda {\textstyle\frac{1}{\chi ^{2} }} \left(K\log \left({10n/K} \right)+\log \left({2 \mathord{\left/{\vphantom{2 \varepsilon }}\right.\kern-\nulldelimiterspace} \varepsilon } \right)\right),                                                 
\end{equation} 
where $\delta _{K} =\delta _{K} \left({\rm {\mathcal M}}\right)$ is the $K$-th restricted isometry constant of ${\rm {\mathcal M}}$, which is the smallest $\chi \ge 0$ such that
\begin{equation} \label{104)} 
\left(1-\chi \right)\ell ^{2} \left({\left| S \right\rangle} \right)\le \left(L_{2} \left({\rm {\mathcal M}}{\left| S \right\rangle} \right)\right)^{2} \le \left(1+\chi \right)\ell ^{2} \left({\left| S \right\rangle} \right),                       
\end{equation} 
for $\forall S$, with $L_{0} \left(S\right)\le K$. Note, that at $\varepsilon =2\exp \left(-{\textstyle\frac{R}{\chi ^{2} 2\lambda }} \right)$, \eqref{ZEqnNum898189} picks up the value of $R=2\lambda {\textstyle\frac{1}{\chi ^{2} }} \left(K\log \left({\textstyle\frac{10n}{K}} \right)\right)$. 

Then, some argumentation on the $L_{1} $-minimization based recovery via basis pursuit in the ${\rm {\mathcal P}}$ unit, yields a condition for the $2K$-th restricted isometry constant, $\delta _{2K} $ of ${\rm {\mathcal M}}$ as
\begin{equation} \label{ZEqnNum306145} 
\delta _{2K} <{\textstyle\frac{1}{3}} .                                                                        
\end{equation} 
The condition in \eqref{ZEqnNum306145}  allows to determine any $\tilde{S}$ in the ${\rm {\mathcal P}}$ post-processing as a unique solution of
\begin{equation} \label{ZEqnNum744529} 
\tilde{S}=\mathop{\arg \min }\limits_{S} L_{1} \left(S\right) 
\end{equation} 
subject to 
\begin{equation} \label{107)} 
Y^{R} =\beta _{C}^{R} S,                                                                        
\end{equation} 
with success probability
\begin{equation} \label{108)} 
\Pr (\tilde{S})=1-\varepsilon =1-\xi ,                                                            
\end{equation} 
where $\beta _{C}^{R} $ is as
\begin{equation} \label{109)} 
\beta _{C}^{R} =\left(\beta _{C}^{\left(1\right)} ,\ldots ,\beta _{C}^{\left(R\right)} \right),                                                              
\end{equation} 
such that for every $S$ there exists a unique solution of  \eqref{ZEqnNum744529}. 

Thus, in an $m$-th measurement round, output vector $z_{m} $ is evaluated as
\begin{equation} \label{ZEqnNum894867} 
\begin{split}
   {{z}_{m}}&=\mathcal{P}\left( {{Y}^{\left( m \right)}} \right) \\ 
 & =\tilde{S} \\ 
 & ={{U}_{B}}\tilde{X} \\ 
 & =U( {\vec{\theta }})\tilde{X},  
\end{split}
\end{equation} 
where $Y^{\left(m\right)} $ is given in \eqref{ZEqnNum485240}.

Since \eqref{ZEqnNum306145} puts a strict bound on $\chi $, it allows us to rewrite \eqref{ZEqnNum898189} at a particular 
\begin{equation} \label{111)} 
\xi =\varepsilon  
\end{equation} 
as
\begin{equation} \label{ZEqnNum810292} 
R\left(\xi \right)=c_{1} K\log \left({\textstyle\frac{10n}{K}} \right)+c_{2} \log \left({\textstyle\frac{2}{\xi }} \right),                                              
\end{equation} 
where $c_{1} ,c_{2} >0$ are constants depend only on the $C_{1} ,C_{2} $ sub-Gaussian parameters (see \sref{subg}) of ${\rm {\mathcal M}}$. 

Therefore, the global optimal $z^{*} $ can be determined via $R$ rounds via ${\rm {\mathcal P}}$, without loss of generality as
\begin{equation} \label{ZEqnNum614290} 
R={\rm {\mathcal O}}\left(\gamma K\log \left({\textstyle\frac{10n}{K}} \right)\right),                                                        
\end{equation} 
where $\gamma >0$ is a constant, that yields $C\left(z^{*} \right)$ via \eqref{ZEqnNum742810}. Thus, at $\gamma =1$, the success probability is
\begin{equation} \label{114)} 
\begin{split}
   \Pr \left( {{z}^{*}} \right)&=\Pr \left( C\left( {{z}^{*}} \right) \right) \\ 
 & =1-2\exp \left( -R \right)=1-\xi ,  
\end{split}
\end{equation} 
that concludes the proof.
\end{proof}

The steps of the dense measurement for an computational basis quantum states are summarized in Procedure 2.

\begin{proced}
  \DontPrintSemicolon
\caption{Dense measurements at $U(\vec{\theta })=U_{B} $}

\textbf{Step 1}. Set the superposed input system ${\left| X \right\rangle} $ \eqref{ZEqnNum910801} and unitary sequence $U(\vec{\theta })$ \eqref{ZEqnNum489626} of $QG$.

\textbf{Step 2}. Select a computational basis $B$ for unitary $U_{B} $ to set ${\left| S \right\rangle} =U_{B} {\left| X \right\rangle} $, such that $L_{0} \left(S\right)\le K$ holds for the $L_{0} $-norm of $S=BX$.

\textbf{Step 3}. For a unitary $U_{B} $, set the gate parameter vector $\vec{\theta }$ in $U(\vec{\theta })$, such that $U(\vec{\theta })=U_{B} $.

\textbf{Step 4}. Measure ${\left| S \right\rangle} =U_{B} {\left| X \right\rangle} $ via $M_{r} $ to get $Y=M_{r} \left({\left| S \right\rangle} \right)$.

\textbf{Step 5}. Apply ${\rm {\mathcal P}}$ post-processing to determine $\tilde{S}$ via \eqref{ZEqnNum744529}. 

\textbf{Step 6}. Output $z$ as in \eqref{ZEqnNum894867}.

\textbf{Step 7}. Apply steps 1-6 through $R$ rounds \eqref{ZEqnNum602021}, to achieve error probability $\xi $ \eqref{ZEqnNum846366} in the estimation of the global optimal output $z^{*} $ and objective function $C\left(z^{*} \right)$.
\end{proced} 

\section{Performance Evaluation}
\label{sec5}
In this section, we analyze the $\Pr \left(C\left(z^{*} \right)\right)=\Pr \left(z^{*} \right)$ success probabilities of finding the global optimal output $z^{*} $, and global optimal estimate $C\left(z^{*} \right)$ in function of $R$, for an arbitrary objective function of the quantum computer with a $QG$ quantum circuit, and arbitrary objective function $C$. First the $U(\vec{\theta })=U_{B} U(\vec{\theta '})$ setting is discussed, then the $U(\vec{\theta })=U_{B} $ situation is proposed.

In \fref{fig3}, a dense measurement at the $U(\vec{\theta })=U_{B} U(\vec{\theta '})$ case is depicted. In this case, $R$ is evaluated as given in \eqref{ZEqnNum602021}. In \fref{fig3}(a) the length of the measured quantum system is fixed to $n=1000$, while $K$ varies between $5$ and $20$. In \fref{fig3}(b), the value of $K$ is fixed to $K=10$, while $n$ varies between $n=10^{1} $ and $n=10^{6} $. For a comparison the results of $R_{0} $ standard measurements \cite{ref10} are also depicted in both figures with dashed gray lines ($R_{0} =100$, $\Pr _{R_{0} } \left(C\left(z^{*} \right)\right)=\Pr _{R_{0} } \left(z^{*} \right)\ge 0.01$).

\begin{center}
\begin{figure*}[!h]
\begin{center}
\includegraphics[angle = 0,width=1\linewidth]{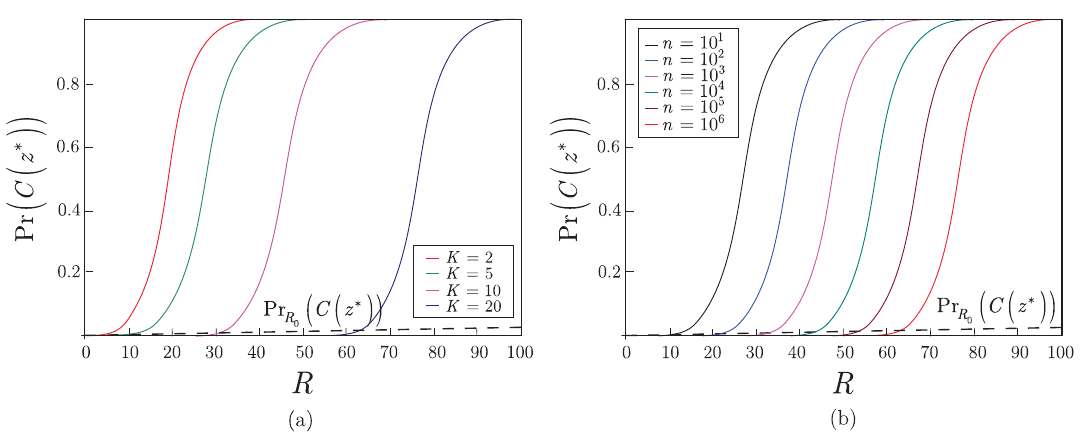}
\caption{(a): The $\Pr \left(C\left(z^{*} \right)\right)$ success probabilities for $U(\vec{\theta })=U_{B} U(\vec{\theta '})$, in function of the dense measurement rounds $R$, $R=\alpha Z^{2} K\log ^{4} \left(n\right)$, at $n=1000$, and $K=2,5,10,20$, $\alpha _{K=2} =2.47\times 10^{-4} $, $\alpha _{K=5} =1.36\times 10^{-4} $, $\alpha _{K=10} =8.46\times 10^{-5} $, $\alpha _{K=20} =6.17\times 10^{-5} $, and $Z=\sqrt{n} $. (b): The $\Pr \left(C\left(z^{*} \right)\right)$ success probabilities for $U(\vec{\theta })=U_{B} U(\vec{\theta '})$, in function of $R$ at $K=10$, for $n=10^{1} ,10^{2} ,10^{3} ,10^{4} ,10^{5} ,10^{6} $.} 
 \label{fig3}
 \end{center}
\end{figure*}
\end{center} 

In \fref{fig4}, a $U(\vec{\theta })=U_{B} $ situation is depicted. In this case, a computational basis quantum state is outputted by the $QG$ structure, and $R$ is evaluated as given in \eqref{ZEqnNum614290}. In \fref{fig4}(a) the length of the measured quantum system is fixed to $n=1000$, while $K$ varies between $5$ and $20$. In \fref{fig4}(b), the value of $K$ is fixed to $K=10$, while $n$ varies between $n=10^{1} $ and $n=10^{6}$.

\begin{center}
\begin{figure*}[!h]
\begin{center}
\includegraphics[angle = 0,width=1\linewidth]{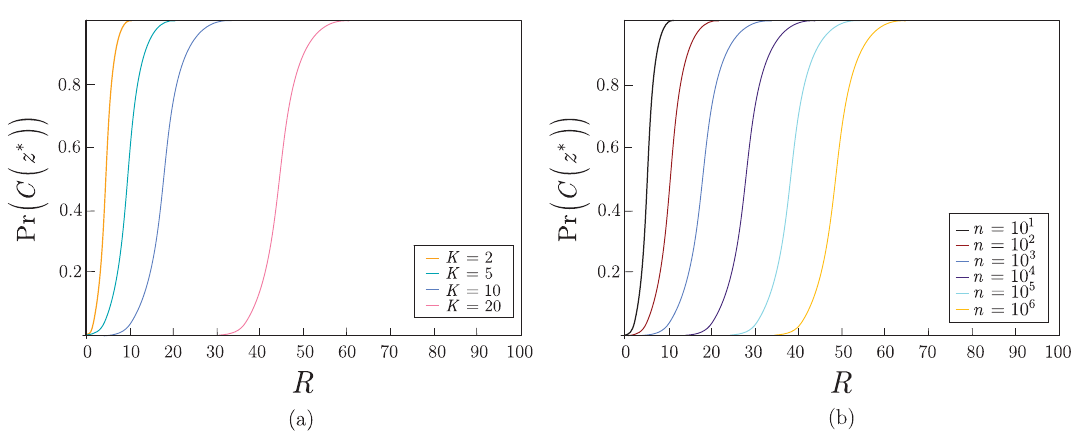}
\caption{(a): The $\Pr \left(C\left(z^{*} \right)\right)$ success probabilities for computational basis quantum states, $U(\vec{\theta })=U_{B} $, in function of the measurement rounds $R$, $R=\gamma K\log \left({\textstyle\frac{10n}{K}} \right)$, at $n=1000$, $K=2,5,10,20$ and $\gamma =1$. (b): The $\Pr \left(C\left(z^{*} \right)\right)$ success probabilities for computational basis quantum states, $U(\vec{\theta })=U_{B} $, in function of $R$ at $K=10$, for $n=10^{1} ,10^{2} ,10^{3} ,10^{4} ,10^{5} ,10^{6} $.} 
 \label{fig4}
 \end{center}
\end{figure*}
\end{center}

\section{Conclusions}
\label{sec6}
Here, we defined a novel measurement technique called dense measurement for quantum computation. Dense measurement utilizes a random measurement strategy and a post-processing unit to eliminate the main drawbacks of standard measurement techniques. The dense measurement method provides two fundamental results. First, it significantly increases the success probability of finding a global optimal measurement result. Second, it radically reduces the number of measurement rounds required to determine a global optimal measurement result. We demonstrated the results through an application of dense measurements with quantum circuits that realize arbitrary unitary operations. We proved the results of dense measurement theory for the measurement of arbitrary quantum states and for the measurement of computational basis quantum states in gate-model quantum computer environment.

\section*{Acknowledgements}
This work was partially supported by the National Research Development and Innovation Office of Hungary (Project No. 2017-1.2.1-NKP-2017-00001), by the Hungarian Scientific Research Fund - OTKA K-112125 and in part by the BME Artificial Intelligence FIKP grant of EMMI (BME FIKP-MI/SC).


\newpage
\appendix
\setcounter{table}{0}
\setcounter{figure}{0}
\setcounter{equation}{0}
\setcounter{algocf}{0}
\renewcommand{\thetable}{\Alph{section}.\arabic{table}}
\renewcommand{\thefigure}{\Alph{section}.\arabic{figure}}
\renewcommand{\theequation}{\Alph{section}.\arabic{equation}}
\renewcommand{\thealgocf}{\Alph{section}.\arabic{algocf}}

\setlength{\arrayrulewidth}{0.1mm}
\setlength{\tabcolsep}{5pt}
\renewcommand{\arraystretch}{1.5}
\section{Appendix}
\subsection{Abbreviations}
\begin{description}
\item[NISQ] Noisy Intermediate-Scale Quantum
\item[QADI] Quantum Adiabatic Algorithm
\item[QAOA] Quantum Approximate Optimization Algorithm
\item[QG] Quantum Gate structure of a quantum circuit
\item[POVM] Positive-Operator Valued Measure
\end{description}

\subsection{Notations}
The notations of the manuscript are summarized in  \tref{tab2}.
\begin{center}
\begin{longtable}{||l|p{4.5in}||}
\caption{Summary of notations.}
\label{tab2}
\endfirsthead
\endhead
\hline
\textit{Notation} & \textit{Description} \\ \hline
$QG$  & Quantum  gate structure of a quantum circuit. \\ \hline 
$M$ & Standard measurement operator. \\ \hline 
$M_{r} $ & Measurement operator in the dense measurement procedure. \\ \hline 
$d$ & Dimension of the quantum system. \\ \hline 
$R_{0} $ & Measurement rounds at a standard measurement $M$. \\ \hline 
$R$ & Measurement rounds of dense measurements, $M_{r} =\left(b_{1} M_{B} ,\ldots ,b_{n} M_{B} \right)^{T} $, where $b_{i} $ is a random variable, $b_{i} \in \left\{0,1\right\}$, ${\rm Pr}\left(0\right)={\rm Pr}\left(1\right)=0.5$, associated with the measurement of the $i$-th quantum state of the output quantum system, while $M_{B} $ is a quantum measurement in the computational basis $B$, $b_{i} M_{B} =0$ if $b_{i} =0$, and $b_{i} M_{B} =M_{B} $ if $b_{i} =1$. \\ \hline 
$z$ & Measurement output. \\ \hline 
$z^{*} $ & Global optimal measurement output. \\ \hline 
$L$ & Number of unitary gates in the $QG$ quantum circuit. \\ \hline 
$U_{i} \left(\theta _{i} \right)$ & An $i$-th unitary gate, $U_{i} \left(\theta _{i} \right)=\exp \left(-i\theta _{i} P\right)$, where $P$ is a generalized Pauli operator formulated by a tensor product of Pauli operators $\left\{\sigma _{X} ,\sigma _{Y} ,\sigma _{Z} \right\}$, while $\theta _{i} $ is referred to as the gate parameter associated to $U_{i} \left(\theta _{i} \right)$. \\ \hline 
$U(\vec{\theta })$ & A unitary operation realized via the quantum circuit, $U(\vec{\theta })=U_{L} \left(\theta _{L} \right)U_{L-1} \left(\theta _{L-1} \right)\ldots U_{1} \left(\theta _{1} \right)$, where $U_{i} \left(\theta _{i} \right)$ identifies an $i$-th unitary gate. \\ \hline 
$\vec{\theta }$ & Gate parameter vector, a collection of gate parameters of the $L$ unitaries, $\vec{\theta }=\left(\theta _{1} ,\ldots ,\theta _{L-1} ,\theta _{L} \right)^{T} $. \\ \hline 
$C$ & Classical objective function of a computational problem fed into the quantum computer.  \\ \hline 
$P$ & Generalized Pauli operator formulated by the tensor product of Pauli operators $\left\{\sigma _{X} ,\sigma _{Y} ,\sigma _{Z} \right\}$. \\ \hline 
${\left| X \right\rangle} $ & An $n$-length input quantum system. \\ \hline 
$X$ & Classical representation of ${\left| X \right\rangle} $. \\ \hline 
${\left| Y \right\rangle} $ & An $n$-length output quantum system, ${\left| Y \right\rangle} =U (\vec{\theta }){\left| X \right\rangle} $. \\ \hline 
$Y$ & An $n$-dimensional output vector. \\ \hline 
$L_{0} $ & $L_{0} $-norm, the number of nonzero elements of $x$, $L_{0} \left(x\right)=\sum _{i=1}^{n}1_{x_{i} \ne 0}  $, or $L_{0} \left(x\right)=\mathop{\lim }\limits_{p\to 0} \sqrt[{p}]{\sum _{i=1}^{n}x_{i}^{p}  } $.   \\ \hline 
$L_{1} $ & $L_{1} $-norm, vector norm, $L_{1} \left(x\right)=\sum _{i=1}^{n}\left|x_{i} \right| $. \\ \hline 
$L_{2} $ & $L_{2} $-norm, Euclidean norm, $L_{2} \left(x\right)=\sqrt{\sum _{i=1}^{n}\left|x_{i} \right|^{2}  } $. \\ \hline 
$\ell ^{2} $ & The $\ell ^{2} $-norm of a quantum system, $\ell ^{2} \left({\left| \psi  \right\rangle} \right)=\sqrt{\sum _{x}\left|\psi \left(x\right)\right|^{2}  } =1$, where $\left|\psi \left(x\right)\right|^{2} =\Pr \left(x\right)$, and  $\sqrt{\int \Pr \left(x\right)dx } =\sqrt{\int \left|\psi \left(x\right)\right|^{2} dx } =1$. \\ \hline 
$U(\vec{\theta '})$ & An actual setting of the unitaries of $QG$ at a particular computational basis $B$, to provide output ${\left| G \right\rangle} =U(\vec{\theta '}){\left| S \right\rangle} $, such that $U(\vec{\theta '}){\left| S \right\rangle} =U(\vec{\theta }){\left| X \right\rangle} $. \\ \hline 
$\vec{\theta '}$ & $L$-dimensional vector of the gate parameters of $U(\vec{\theta '})$. \\ \hline 
$B$  & Computational basis, selected such that $L_{0} \left(S\right)=K$, $K{\rm \ll }n$, holds for the $L_{0} $-norm of $S$, where $S$ is a classical representation of ${\left| S \right\rangle} $. \\ \hline 
$U_{B} $ & Unitary that sets computational basis $B$ as $U_{B} {\left| X \right\rangle} ={\left| S \right\rangle} $.  \\ \hline 
${\rm {\mathcal P}}$ & Post processing unit to perform an $L_{1} $-minimization, and post-processing calculations. \\ \hline 
$C\left(z^{*} \right)$ & Optimal estimate of a particular objective function $C$ fed into the quantum circuit, $C\left(z^{*} \right)=\mathop{\max }\limits_{\forall m} C\left(z_{m} \right)$, where $C\left(z_{m} \right)$ is the estimate yielded in an $m$-th measurement round, $m=1,\ldots ,R_{0} $, while $z_{m} $ is the output string yielded in the $m$-th round.   \\ \hline 
$\Pr _{R_{0} } \left(z^{*} \right)$ & Probability of finding the global optimal output $z^{*} $ via $R_{0} $ standard measurement rounds.  \\ \hline 
$\Pr _{R} \left(z^{*} \right)$ & Probability of finding the global optimal $z^{*} $ via $R$ dense measurement rounds. \\ \hline 
$\Pr _{R_{0} } C\left(z^{*} \right)$ & Probability of finding the global optimal $C\left(z^{*} \right)$ via $R_{0} $ standard measurement rounds. \\ \hline 
$\Pr _{R} C\left(z^{*} \right)$ & Probability of finding the global optimal $C\left(z^{*} \right)$ via $R$ dense measurement rounds. \\ \hline 
$K$ & Constant, $L_{0} \left(S\right)\le K$, $K{\rm \ll }n$. \\ \hline 
$\varepsilon $ & Probability, $\varepsilon =\Pr \left(\delta _{K} \ge \chi \right)$, where $\delta _{K} $ is a constant. \\ \hline 
$M_{r}^{\left(m\right)} $ & Measurement operator of the $m$-th, $m=1,\ldots ,R$, dense measurement round $M_{r}^{\left(m\right)} =\left(b_{1}^{\left(m\right)} M_{B} ,\ldots ,b_{n}^{\left(m\right)} M_{B} \right)^{T} $. \\ \hline 
${\rm {\mathcal M}}$ & Measurement matrix formulated via $R$ dense measurement rounds, ${\rm {\mathcal M}}=\left(M_{r}^{\left(1\right)} ,\ldots ,M_{r}^{\left(R\right)} \right)$. \\ \hline 
${\rm {\mathcal Q}}$ & Matrix, ${\rm {\mathcal Q}}={\rm {\mathcal M}}U(\vec{\theta '})$. \\ \hline 
${\left| S \right\rangle} $ & Computational basis quantum state, ${\left| S \right\rangle} =U_{B} {\left| X \right\rangle} $. \\ \hline 
${\left| G \right\rangle} $ & An output quantum system, ${\left| G \right\rangle} =U(\vec{\theta '}){\left| S \right\rangle} $. \\ \hline 
$b_{i} $ & A random variable, \newline $b_{i} =\left\{\begin{array}{c} {0,{\rm \; with\; Pr}\left(0\right)=0.5} \\ {1,{\rm \; with\; Pr}\left(1\right)=0.5} \end{array}\right. ,$ \newline associated with the measurement of the $i$-th quantum system. \\ \hline 
$Y$ & An $n$-bit length output vector. \\ \hline 
$M'_{r} $ & A measurement operator, $M'_{r} =M_{r} U(\vec{\theta '})$. \\ \hline 
$\beta _{C} $ & An $n$-length vector, $\beta _{C} =\left(b_{1} ,\ldots ,b_{n} \right)^{T} $. \\ \hline 
$\beta '_{C} $ & An $n$-length vector, $\beta '_{C} =\beta _{C} U(\vec{\theta '})$. \\ \hline 
$\Lambda $ & A parameter, $\Lambda =U(\vec{\theta '})S$. \\ \hline 
$\tilde{S}$ & Recovered computational basis vector $S$ from $Y=\beta '_{C} S$  via ${\rm {\mathcal P}}$ . \\ \hline 
$\tilde{\Lambda }$ & An optimal value of $\Lambda $ evaluated from $\tilde{S}$, $\tilde{\Lambda }=U(\vec{\theta '})\tilde{S}$. \\ \hline 
$Y^{R} $ & Measurement output matrix of $R$ measurement rounds, $Y^{R} ={\rm {\mathcal Q}}{\left| S \right\rangle} =\left(Y^{\left(1\right)} ,\ldots ,Y^{\left(R\right)} \right)$, where $Y^{\left(m\right)} $ is the measurement result vector of the $m$-th round. \\ \hline 
$\xi $ & Error probability of finding $z^{*} $ at the end of the $R$ rounds, $\Pr \left(z\ne z^{*} \right)=\xi $. \\ \hline 
$C_{1} ,C_{2} $ & Sub-Gaussian parameters, $C_{1} ,C_{2} >0$. \\ \hline 
$\delta _{K} $ & $K$-th restricted isometry constant. \\ \hline 
$\left[n\right]$ & Set of natural numbers not exceeding $n$, $\left[n\right]=\left\{1,\ldots ,n\right\}$. \\ \hline 
$\Upsilon $ & A subset. \\ \hline 
$H$ & Hermitian matrix. \\ \hline 
$\Omega $, $\kappa $, $\chi $ & Parameters of the dense measurement procedure. \\ \hline 
$V,W,D$ & Parameters of the dense measurement procedure. \\ \hline 
${\rm {\mathcal B}}_{\Upsilon } $ & Unit ball. \\ \hline 
$\Gamma $ & A finite subset of ${\rm {\mathcal B}}_{\Upsilon } $. \\ \hline 
$\delta _{K} <\chi $ & Event associated with probability $\Pr \left(\delta _{K} <\chi \right)=1-\varepsilon $, where $\delta _{K} $ is the $K$-th restricted isometry constant. \\ \hline 
$U(\vec{\theta }_{q,k})$ & A $q$-th element of the $k$-th column of $U(\vec{\theta })$. \\ \hline 
$Z$ & Constant, $Z\ge \sqrt{n} \mathop{\max }\limits_{k,q\in \left[n\right]} \left|U(\vec{\theta }_{q,k})\right|$. \\ \hline 
$\alpha $ & Constant, $\alpha >0$. \\ \hline 
$u_{k} $ & A $k$-th column  of $U(\vec{\theta })$, $k=1,\ldots ,n$. \\ \hline 
$v_{k} $ & A normalized $k$-th column of $U(\vec{\theta })$, $v_{k} =\sqrt{n} u_{k} $, $k=1,\ldots ,n$. \\ \hline 
$\varphi _{kl} $ & Inner product of two normalized columns $v_{k} $ and $v_{l} $, as $\varphi _{kl} =\left\langle {\textstyle\frac{1}{\sqrt{n} }} v_{k} ,{\textstyle\frac{1}{\sqrt{n} }} v_{l} \right\rangle =\left\langle u_{k} ,u_{l} \right\rangle $. \\ \hline 
$u_{i,j} $ & A unitary $u_{i,j} =U(\vec{\theta }_{j,i})$, $j$-th element of the $i$-th column of $U(\vec{\theta })$.  \\ \hline 
$v_{i,j} $ & A normalization of $u_{i,j} $, $v_{i,j} =\sqrt{n} U(\vec{\theta }_{j,i})$. \\ \hline 
$u'_{k} $ & A $k$-th column of $U(\vec{\theta '})$. \\ \hline 
$b_{q} $ & A $q$-th column of  $U_{B} $. \\ \hline 
${\rm {\mathcal P}}_{Q_{R} } $ & Projector, selects a subset of $U(\vec{\theta })$ in the $R$ rounds, where $Q_{R} \subset \left[n\right]$ is a subset of $R$ elements selected uniform at random from all subsets of $\left[n\right]$ of cardinality $R$, $\left|Q_{R} \right|=R$. \\ \hline 
$\xi ^{*} $ & Error probability associated with the selection of rows uniformly and independently at random from $U(\vec{\theta })$. \\ \hline 
$Q_{R} $ & A subset of $R$ elements selected uniform at random from all subsets of $\left[n\right]$ of cardinality $R$, $Q_{R} \subset \left[n\right]$, $\left|Q_{R} \right|=R$. \\ \hline 
$Q'_{R} $ & A subset of $R$ elements,  elements are selected independently and uniformly at random from $\left[n\right]$, $Q'_{R} \subset \left[n\right]$, $\left|Q'_{R} \right|=R$. \\ \hline 
$Q_{k} $ & A subset of $k\le R$ selected uniform at random from all subsets of $\left[n\right]$ of cardinality $k$, $Q_{k} \subset \left[n\right]$, $\left|Q_{k} \right|=k$. \\ \hline 
${\rm {\mathcal E}}\left(Q\right)$ & Event that the $L_{1} $-minimization in ${\rm {\mathcal P}}$ fails. \\ \hline 
${\rm {\mathcal D}}\left(\cdot \right)$ & A distribution. \\ \hline 
$\Lambda ^{*} $ & An optimal $\tilde{\Lambda }$ determined via ${\rm {\mathcal P}}$. \\ \hline 
$\delta _{2K} $ & $2K$-th restricted isometry constant. \\ \hline
\end{longtable}
\end{center}

\begin{thebibliography}{10}
\bibitem{refpr} Preskill, J. Quantum Computing in the NISQ era and beyond, \textit{Quantum} 2, 79 (2018).

\bibitem{refha} Harrow, A. W. and Montanaro, A. Quantum Computational Supremacy, \textit{Nature}, vol 549, pages 203-209 (2017).

\bibitem{aar} Aaronson, S. and Chen, L. Complexity-theoretic foundations of quantum supremacy experiments. \textit{Proceedings of the 32nd Computational Complexity Conference}, CCC '17, pages 22:1-22:67, (2017).

\bibitem{ref1} Debnath, S. et al. Demonstration of a small programmable quantum computer with atomic qubits. \textit{Nature} 536, 63-66 (2016). 

\bibitem{ref2} Barends, R. et al. Superconducting quantum circuits at the surface code threshold for fault tolerance. \textit{Nature} 508, 500-503 (2014).

\bibitem{ref3} Ofek, N. et al. Extending the lifetime of a quantum bit with error correction in superconducting circuits. \textit{Nature} 536, 441-445 (2016).

\bibitem{ref4} Kielpinski, D., Monroe, C. and Wineland, D. J. Architecture for a large-scale ion-trap quantum computer. \textit{Nature} 417, 709-711 (2002).

\bibitem{ref5} Biamonte, J. et al. Quantum Machine Learning. \textit{Nature}, 549, 195-202 (2017).
 
\bibitem{ref6} LeCun, Y., Bengio, Y. and Hinton, G. Deep Learning. \textit{Nature} 521, 436-444 (2014).

\bibitem{ref7} Monz, T. et al. Realization of a scalable Shor algorithm. \textit{Science} 351, 1068-1070 (2016).

\bibitem{ref8} Goodfellow, I., Bengio, Y. and Courville, A. \textit{Deep Learning}. MIT Press. Cambridge, MA, 2016.

\bibitem{ref9} Farhi, E., Goldstone, J. and Gutmann, S. A Quantum Approximate Optimization Algorithm. \textit{arXiv:1411.4028. }(2014).

\bibitem{ref10} Farhi, E., Goldstone, J., Gutmann, S. and Neven, H. Quantum Algorithms for Fixed Qubit Architectures. \textit{arXiv:1703.06199v1} (2017).

\bibitem{ref11} Farhi, E. and Neven, H. Classification with Quantum Neural Networks on Near Term Processors, \textit{arXiv:1802.06002v1} (2018).

\bibitem{ref12} Farhi, E., Goldstone, J. and Gutmann, S. A Quantum Approximate Optimization Algorithm Applied to a Bounded Occurrence Constraint Problem. \textit{arXiv:1412.6062}. (2014).

\bibitem{su} Farhi, E. and Harrow, A. W. Quantum Supremacy through the Quantum Approximate Optimization Algorithm, \textit{arXiv:1602.07674} (2016).

\bibitem{sat} Farhi, E., Kimmel, S. and Temme, K. A Quantum Version of Schoning's Algorithm Applied to Quantum 2-SAT, \textit{arXiv:1603.06985} (2016).

\bibitem{adi1} Farhi, E., Goldstone, J., Gutmann, S. and Sipser, M. Quantum computation by adiabatic evolution. \textit{Technical Report MIT-CTP-2936}, MIT, arXiv:quant-ph/0001106 (2000).

\bibitem{adi2} Kadowaki, T. and Nishimori, H. Quantum annealing in the transverse Ising model. \textit{Phys. Rev.E}, 58:5355-5363, arXiv:cond-mat/9804280 (1998).

\bibitem{refibm} IBM. \textit{A new way of thinking: The IBM quantum experience}. URL: http://www.research.ibm.com/quantum. (2017).

\bibitem{sch} Schoning, T. A probabilistic algorithm for $k$-SAT and constraint satisfaction problems. \textit{Foundations of Computer Science}, 1999. 40th Annual Symposium on, pages 410–414. IEEE (1999).

\bibitem{ref13} Rebentrost, P., Mohseni, M. and Lloyd, S. Quantum Support Vector Machine for Big Data Classification. \textit{Phys. Rev. Lett.} 113. (2014).

\bibitem{ref14} Lloyd, S. The Universe as Quantum Computer, \textit{A Computable Universe: Understanding and exploring Nature as computation}, H. Zenil ed., World Scientific, Singapore, 2012, \textit{arXiv:1312.4455v1} (2013).

\bibitem{ref15} Lloyd, S., Mohseni, M. and Rebentrost, P. Quantum algorithms for supervised and unsupervised machine learning, \textit{arXiv:1307.0411v2} (2013).

\bibitem{ref16} Lloyd, S., Garnerone, S. and Zanardi, P. Quantum algorithms for topological and geometric analysis of data. \textit{Nat. Commun}., 7, arXiv:1408.3106 (2016).

\bibitem{ref17} Lloyd, S., Shapiro, J. H., Wong, F. N. C., Kumar, P., Shahriar, S. M. and Yuen, H. P. Infrastructure for the quantum Internet. \textit{ACM SIGCOMM Computer Communication Review}, 34, 9-20 (2004).

\bibitem{ref18} Lloyd, S., Mohseni, M. and Rebentrost, P. Quantum principal component analysis. \textit{Nature Physics}, 10, 631 (2014).

\bibitem{ref19} Gyongyosi, L., Imre, S. and Nguyen, H. V. A Survey on Quantum Channel Capacities, \textit{IEEE Communications Surveys and Tutorials} \textbf{99}, 1, DOI: 10.1109/COMST.2017.2786748 (2018).

\bibitem{refsur} Gyongyosi, L. and Imre, S. A Survey on Quantum Computing Technology, \textit{Computer Science Review}, Elsevier, DOI: 10.1016/j.cosrev.2018.11.002, ISSN: 1574-0137, (2018).

\bibitem{m1} Neumann, J. \textit{Mathematical Foundations of Quantum Mechanics} (New ed.). ISBN 9781400889921 (2018).

\bibitem{m2} Wheeler, J. A. and Zurek, W. H. (eds.) \textit{Quantum Theory and Measurement}, Princeton University Press. ISBN 978-0-691-08316-2 (1983).

\bibitem{m3} Braginsky, V. B. and Khalili, F. Y. \textit{Quantum Measurement}, Cambridge University Press. ISBN 978-0-521-41928-4 (1992).

\bibitem{m3b} Jacobs, K. \textit{Quantum Measurement Theory and its Applications}, ISBN-10: 1107025486, ISBN-13: 978-1107025486, Cambridge University Press (2014).

\bibitem{m4} Zurek, W. H. Decoherence, einselection, and the quantum origins of the classical, \textit{Reviews of Modern Physics} 75, 715 (2003). 

\bibitem{m5} Greenstein, G. S. and Zajonc, A. G. \textit{The Quantum Challenge: Modern Research On The Foundations Of Quantum Mechanics} (2nd ed.). ISBN 978-0763724702 (2006).

\bibitem{m6} Jaeger, G. Quantum randomness and unpredictability, \textit{Philosophical Transactions of the Royal Society of London A}, DOI: 10.1002/prop.201600053 (2016)

\bibitem{m7} Jabs, A. A conjecture concerning determinism, reduction, and measurement in quantum mechanics, \textit{Quantum Stud.: Math. Found.} 3:279-292, DOI 10.1007/s40509-016-0077-7 (2016).

\bibitem{m8} Helstrom, C. W. \textit{Quantum Detection and Estimation Theory}. Academic Press, Inc. ISBN 0123400503 (1976).

\bibitem{m9} Barnett, S. M., Pegg, D. T. and Jeffers, J. Bayes' theorem and quantum retrodiction, \textit{J. Mod. Opt.} 47, 1779 (2000).

\bibitem{m10} Amri, T. Quantum behavior of measurement apparatus, \textit{arXiv:1001.3032} (2010).

\bibitem{p1} Preskill, J. Lecture Notes for Physics: Quantum Information and Computation, web:  http://www.theory.caltech.edu/people/preskill/ph229\#lecture

\bibitem{p2} Davies, E. B. \textit{Quantum Theory of Open Systems}, Academic Press (1976).

\bibitem{p3} Holevo, A. S. \textit{Probabilistic and statistical aspects of quantum theory}, North-Holland Publ. Cy., Amsterdam (1982).

\bibitem{p4} Kraus, K. States, Effects, and Operations, \textit{Lecture Notes in Physics 190}, Springer (1983).

\bibitem{p5} Nielsen, M. and Chuang, I. \textit{Quantum Computation and Quantum Information}, Cambridge University Press, (2010).

\bibitem{n1} Gelfand, I. M. and Neumark, M. A. On the embedding of normed rings into the ring of operators in Hilbert space, \textit{Rec. Math.} N.S. 12(54) 197–213 (1943).

\bibitem{n2} Peres, A. Neumark’s theorem and quantum inseparability. \textit{Foundations of Physics}, 12:1441–1453, (1990).

\bibitem{n3} Peres, A. \textit{Quantum Theory: Concepts and Methods}, Kluwer Academic Publishers (1993).

\bibitem{u1} Ivanovic, I. D., How to differentiate between non-orthogonal states, \textit{Phys. Lett. A} 123, 257-259 (1987).

\bibitem{u2} Dieks, D. Overlap and distinguishability of quantum states, \textit{Phys. Lett. A} 126 303 (1988).

\bibitem{u3} Peres, A. How to differentiate between non-orthogonal states, \textit{Phys. Lett. A} 128 19 (1988).

\bibitem{u4} Chefles, A. Quantum State Discrimination, \textit{Contemp. Phys.} 41, 401 (2000).
 
\bibitem{u5} Bergou, J. A., Herzog, U. and Hillery, M. Discrimination of Quantum States, \textit{Lect. Notes Phys.} 649, 417–465 (2004).

\bibitem{ref20} Donoho, D. Compressed Sensing, \textit{IEEE Trans. Info. Theory}, vol. 52, no. 4, pp. 1289–1306 (2006).

\bibitem{ref21} Candes, E. J. and Tao, T. Near optimal signal recovery from random projections: Universal encoding strategies?,  \textit{IEEE Trans. Info. Theory}, vol. 52, no. 12, pp. 5406–5425 (2006).

\bibitem{ref22} Candes, E. J., Romberg, J. and Tao, T. Robust Uncertainty Principles: Exact Signal Reconstruction from Highly Incomplete Frequency Information, \textit{IEEE Trans. Info. Theory}, vol. 52, no. 2, pp. 489–509 (2006).

\bibitem{ref23} Foucart, S. and Rauhut, H. \textit{A Mathematical Introduction to Compressive Sensing}, Springer, ISBN 978-0-8176-4947-0 ISBN 978-0-8176-4948-7, DOI 10.1007/978-0-8176-4948-7 (2013).

\bibitem{refcs1} Eldar, Y. C. and Kutyniok, G. \textit{Compressed sensing: theory and applications}. Cambridge University Press (2012).

\bibitem{refcs2} Candes, E. J. Compressive sampling, in \textit{Int. Congress of Mathematicians}, Madrid, Spain, vol. 3, pp. 1433–1452 (2006).

\bibitem{refcs3} Baraniuk, R. G. Compressive sensing, \textit{IEEE Signal Proc. Mag.}, vol. 24, no. 4, pp. 118–120, 124 (2007).

\bibitem{refcs4} Candes, E. J. and Wakin, M. B. An introduction to compressive sampling, \textit{IEEE Signal Proc. Mag.}, vol. 25, no. 2, pp. 21–30 (2008).

\bibitem{refcs5} Duarte, M. F. and Eldar, Y. C. Structured Compressed Sensing: From Theory to Applications, \textit{IEEE Transactions on Signal Processing}, DOI: 10.1109/TSP.2011.2161982 (2011).

\bibitem{refcs16} Baraniuk, R. G., Davenport, M., DeVore, R. and Wakin, M. B. A simple proof of the restricted isometry property for random matrices, \textit{Constructive Approximation}, vol. 28, no. 3, pp. 253–263, (2008).

\bibitem{refcs17} Candes, E. J. The Restricted Isometry Property and Its Implications for Compressed Sensing, \textit{Acad. Sci. Paris, Ser.} I 346 (2008).

\bibitem{refcs18} Bajwa, W. U., Sayeed, A. and Nowak, R. A restricted isometry property for structurally subsampled unitary matrices, in \textit{Allerton Conf. Communication, Control, and Computing}, Monticello, IL, pp. 1005–1012 (2009).

\bibitem{refcs19} Wakin, M. B. and Davenport, M. A. Analysis of orthogonal matching pursuit using the restricted isometry property, \textit{IEEE Trans. Info. Theory}, vol. 56, no. 9, pp. 4395–4401, (2010).

\bibitem{refcs20} Rauhut, H., Romberg, J. K. and Tropp, J. A. Restricted isometries for partial random circulant matrices, \textit{Appl. Comput. Harmon. Anal.} (2011).

\bibitem{one1} Boufounos, P. T. and Baraniuk, R. G. 1-bit compressive sensing. In \textit{Proceedings of the 42nd Annual Conference on Information Sciences and Systems} (CISS), pages 16–21. IEEE, (2008).

\bibitem{one2} Gopi, S., Netrapalli, P., Jain, P. and Nori, A. One-bit compressed sensing: Provable support and vector recovery. In \textit{Proceedings of the 30th International Conference on Machine Learning} (ICML), pages 154–162, (2013).

\bibitem{one3} Knudson, K., Saab, R. and Ward, R. One-bit compressive sensing with norm estimation. \textit{arXiv:1404.6853} (2014).

\bibitem{refcs6} Rauhut, H., Schnass, K. and Vandergheynst, P. Compressed sensing and redundant dictionaries. \textit{IEEE Trans. Info. Theory}, 54(5):2210–2219, (2008).

\bibitem{refcs7} Duarte, M. F., Davenport, M. A., Takhar, D., Laska, J. N., Sun, T., Kelly, K. F. and Baraniuk, R. G. Single pixel imaging via compressive sampling, \textit{IEEE Signal Proc. Mag.}, vol. 25, no. 2, pp. 83–91 (2008).

\bibitem{refcs8} Rauhut, H. Circulant and Toeplitz matrices in compressed sensing, \textit{arXiv:0902.4394}, (2009).

\bibitem{refcs9} Blumensath, T. Sampling and reconstructing signals from a union of linear subspaces. \textit{IEEE Trans. Info. Theory}, 57(7):4660–4671, (2011).

\bibitem{refcs10} Fan, Y.-Z., Huang, T. and Zhu, M. Compressed Sensing Based on Random Symmetric Bernoulli Matrix, \textit{arXiv:1212.3799} (2012).

\bibitem{refcs11} Baraniuk, R. G., Foucart, S., Needell, D., Plan, Y. and Wootters, M. Exponential decay of reconstruction error from binary measurements of sparse signals. \textit{arXiv:1407.8246}, (2014).

\bibitem{refcs12} Krahmer, F., Needell, D. and Ward. R. Compressive sensing with redundant dictionaries and structured measurements. \textit{SIAM Journal on Mathematical Analysis}, 47(6):4606– 4629 (2015).

\bibitem{refcs13} Foucart, S. Dictionary-sparse recovery via thresholding-based algorithms. \textit{Journal of Fourier Analysis and Applications}, 22(1):6–19 (2016).

\bibitem{ref24} Van Meter, R. \textit{Quantum Networking}, John Wiley and Sons Ltd, ISBN 1118648927, 9781118648926 (2014).

\bibitem{ref25} Imre, S. and Gyongyosi, L. \textit{Advanced Quantum Communications - An Engineering Approach}. Wiley-IEEE Press (New Jersey, USA), (2012).

\bibitem{ref26} Shor, P. W. Scheme for reducing decoherence in quantum computer memory. \textit{Phys. Rev. A}, 52, R2493-R2496 (1995).

\bibitem{ref27} Petz, D. \textit{Quantum Information Theory and Quantum Statistics}, Springer-Verlag, Heidelberg, (2008).

\bibitem{ref28} Bacsardi, L. On the Way to Quantum-Based Satellite Communication, \textit{IEEE Comm. Mag.} 51:(08) pp. 50-55. (2013). 

\bibitem{ref29} Laurenza, R. and Pirandola, S. General bounds for sender-receiver capacities in multipoint quantum communications, \textit{Phys. Rev. A} 96, 032318 (2017).

\bibitem{ref30} Gyongyosi, L. and Imre, S. Entanglement-Gradient Routing for Quantum Networks, \textit{Sci. Rep.}, Nature, DOI:10.1038/s41598-017-14394-w (2017).

\bibitem{ref31} Gyongyosi, L. and Imre, S. Entanglement Availability Differentiation Service for the Quantum Internet, \textit{Sci. Rep.}, Nature, DOI:10.1038/s41598-018-28801-3 (2018). 

\bibitem{ref32} Gyongyosi, L. and Imre, S. Multilayer Optimization for the Quantum Internet, \textit{Sci. Rep.}, Nature, DOI:10.1038/s41598-018-30957-x (2018). 

\bibitem{ref33} Gyongyosi, L. and Imre, S. Decentralized Base-Graph Routing for the Quantum Internet, \textit{Phys. Rev. A}, American Physical Society, DOI: 10.1103/PhysRevA.98.022310 (2018).

\bibitem{refa3} Brandao, F. G. S. L., Broughton, M., Farhi, E., Gutmann, S. and Neven, H. For Fixed Control Parameters the Quantum Approximate Optimization Algorithm's Objective Function Value Concentrates for Typical Instances, \textit{arXiv:1812.04170} (2018).

\bibitem{refa4} Zhou, L., Wang, S.-T., Choi, S., Pichler, H. and Lukin, M. D. Quantum Approximate Optimization Algorithm: Performance, Mechanism, and Implementation on Near-Term Devices, arXiv:1812.01041 (2018).

\bibitem{refa5} Lechner, W. Quantum Approximate Optimization with Parallelizable Gates, \textit{arXiv:1802.01157v2} (2018).


\end{thebibliography}
\end{document}